\documentclass[letterpaper,openany]{article} 
\usepackage{aaai23}  
\usepackage{amsfonts}

\usepackage{times}  
\usepackage{helvet}  
\usepackage{courier}  
\usepackage[hyphens]{url}  
\usepackage{graphicx} 
\usepackage{natbib}  
\usepackage{caption} 
\usepackage{tcolorbox}
\usepackage{algorithm}
\usepackage{algorithmic}
\usepackage{bm}
\usepackage{amssymb}
\usepackage{amsthm}
\usepackage{amsmath}
\usepackage{wrapfig}
\usepackage{caption}
\usepackage{subcaption}
\usepackage{fontenc}
\usepackage{inputenc}
\usepackage{graphicx}
\usepackage{caption,booktabs}
\usepackage{float}
\usepackage{tikz}
\usepackage{mathdots}
\usepackage{yhmath}
\usepackage{cancel}
\usepackage{siunitx}
\usepackage{array}
\usepackage{wrapfig}
\usepackage{enumitem}
\usepackage{multirow}
\usepackage{mathtools}
\usepackage{textcomp}
\usepackage{gensymb}
\usepackage{tabularx}
\usepackage{dsfont}
\usepackage{subcaption}
\usepackage{tikzscale}
\usepackage{colortbl}

\newcolumntype{P}[1]{>{\centering\arraybackslash}p{#1}}

\renewcommand{\qedsymbol}{$\blacksquare$}
 
\newtheorem{theorem}{Theorem}
\newtheorem{lemma}{Lemma}
 
\newtheorem{remark}{Remark}

\newtheorem{definition}{Definition}
\newtheorem{problem}{Problem}

\DeclareMathOperator*{\argmax}{arg\,max}

\csname @openrightfalse\endcsname
\usepackage{newfloat}
\usepackage{listings}
\DeclareCaptionStyle{ruled}{labelfont=normalfont,labelsep=colon,strut=off} 
\floatstyle{ruled}
\newfloat{listing}{tb}{lst}{}
\floatname{listing}{Listing}
\title{Does it pay to optimize AUC?}
\author{
    Baojian Zhou\textsuperscript{\rm 1},
    Steven Skiena\textsuperscript{\rm 2}
}
\affiliations{
    \textsuperscript{\rm 1}Fudan University, Shanghai, \textsuperscript{\rm 2}Stony Brook University, New York\\
    bjzhou@fudan.edu.cn, skiena@cs.stonybrook.edu
}

\usepackage{bibentry}

\begin{document}

\maketitle

\begin{abstract}

The Area Under the ROC Curve (AUC) is an important model metric for evaluating binary classifiers, and many algorithms have been proposed to optimize AUC approximately. It raises the question of whether the generally insignificant gains observed by previous studies are due to inherent limitations of the metric or the inadequate quality of optimization.

To better understand the value of optimizing for AUC, we present an efficient algorithm, namely AUC-opt, to find the provably optimal AUC linear classifier in $\mathbb{R}^2$, which runs in $\mathcal{O}(n_+ n_- \log (n_+ n_-))$ where $n_+$ and $n_-$ are the number of positive and negative samples respectively. Furthermore, it can be naturally extended to $\mathbb{R}^d$ in $\mathcal{O}((n_+n_-)^{d-1}\log (n_+n_-))$ by calling AUC-opt in lower-dimensional spaces recursively. We prove the problem is NP-complete when $d$ is not fixed, reducing from the \textit{open hemisphere problem}.

Experiments show that compared with other methods, AUC-opt achieves statistically significant improvements on between 17 to 40 in $\mathbb{R}^2$ and between 4 to 42 in $\mathbb{R}^3$ of 50 t-SNE training datasets. However, generally the gain proves insignificant on most testing datasets compared to the best standard classifiers. Similar observations are found for nonlinear AUC methods under real-world datasets.
\end{abstract}

\section{Introduction}

The Area Under the ROC  Curve (AUC)~\citep{hanley1982meaning} is an important model evaluation metric that can be applied to a wide range of learning tasks such as binary classification~\citep{bradley1997use}, bipartite ranking~\citep{freund2003efficient}, and recently fairness learning~\citep{kallus2019fairness,vogel2021learning}. It is a generally more reliable quality measure than the accuracy when the dataset is highly imbalanced, which often is the case in real-world problems. Multiple studies~\citep{cortes2004auc,joachims2005support} argue that optimizing classifiers for AUC may result in better classifiers than minimizing error rates.

A wide variety of algorithms \citep{provost2001robust,ferri2002learning,yan2003optimizing,cortes2004auc,rakotomamonjy2004optimizing,herschtal2004optimising,brefeld2005auc,joachims2005support,calders2007efficient,qin2010general,le2010optimization,zhao2011online,gao2013one,ying2016stochastic,eban2017scalable,liu2020stochastic} have been proposed to optimize AUC approximately under different learning settings. Typically, these methods relax the original \textit{nonconvex nondifferentiable} objective to either convex or differentiable. Despite these advances, there exists no strong evidence that these algorithms generally perform better than standard classifiers. Empirical observations \citep{rakotomamonjy2004optimizing,joachims2005support} from previous indicate generally minor and statistically insignificant gains on particular datasets. This vagueness makes the question whether the observed results are due to the inherent limitations of the metric or the inadequate quality of optimization.

To better understand the virtues of optimizing for AUC, we investigate it from both \textit{computational} and \textit{algorithmic} viewpoints. Although AUC optimization is often reported to be NP-hard for linear hypothesis classes \citep{gao2013one,gao2015consistency,gultekin2020mba}, we show that it is polynomial-time solvable if the data dimension $d$ is fixed. We also prove that it is NP-complete if $d$ is not fixed in advance. The key idea of our proof is a reduction from the \textit{open hemisphere problem} \citep{johnson1978densest}. 

With the hope of polynomial-time solvable of linear classifiers, we present an efficient algorithm, namely AUC-opt, that can provably optimize AUC in $\mathbb{R}^2$. A key observation is that given any $n$ training samples on the plane, the number of ``interesting'' classifiers is at most $\mathcal{O}\left(n_+ n_-\right)$ where $n_+$ and $n_-$ are the number of positive and negative samples respectively. Inspired by the idea of the topological sweep \citep{edelsbrunner1989topologically}, we calculate the AUC for the ``minimal'' slope once and then iterate through all other slopes by an ascending order using only constant update-time per iteration, yielding an $\mathcal{O}(n_+ n_-\log n)$ algorithm. Furthermore, our method can be naturally extended to $\mathbb{R}^d$ in $\mathcal{O}((n_+ n_-)^{d-1}\log (n))$ by calling AUC-opt in low-dimensional spaces recursively. This algorithm as an exponential depends on 
$d$ and hence will be impractical on large real-world datasets where samples are usually high-dimensional. \textbf{Our goal here is not a general-purpose algorithm but to address whether the observed limitations of previous AUC optimizers result from inadequate optimization of a convex function or are an inherent result of the AUC objective criteria.} Doing such experiments requires the exact optimization of AUC-opt, even if we are limited to small data sets in low dimensions.

\begin{figure}
\centering
\includegraphics[width=.42\textwidth]{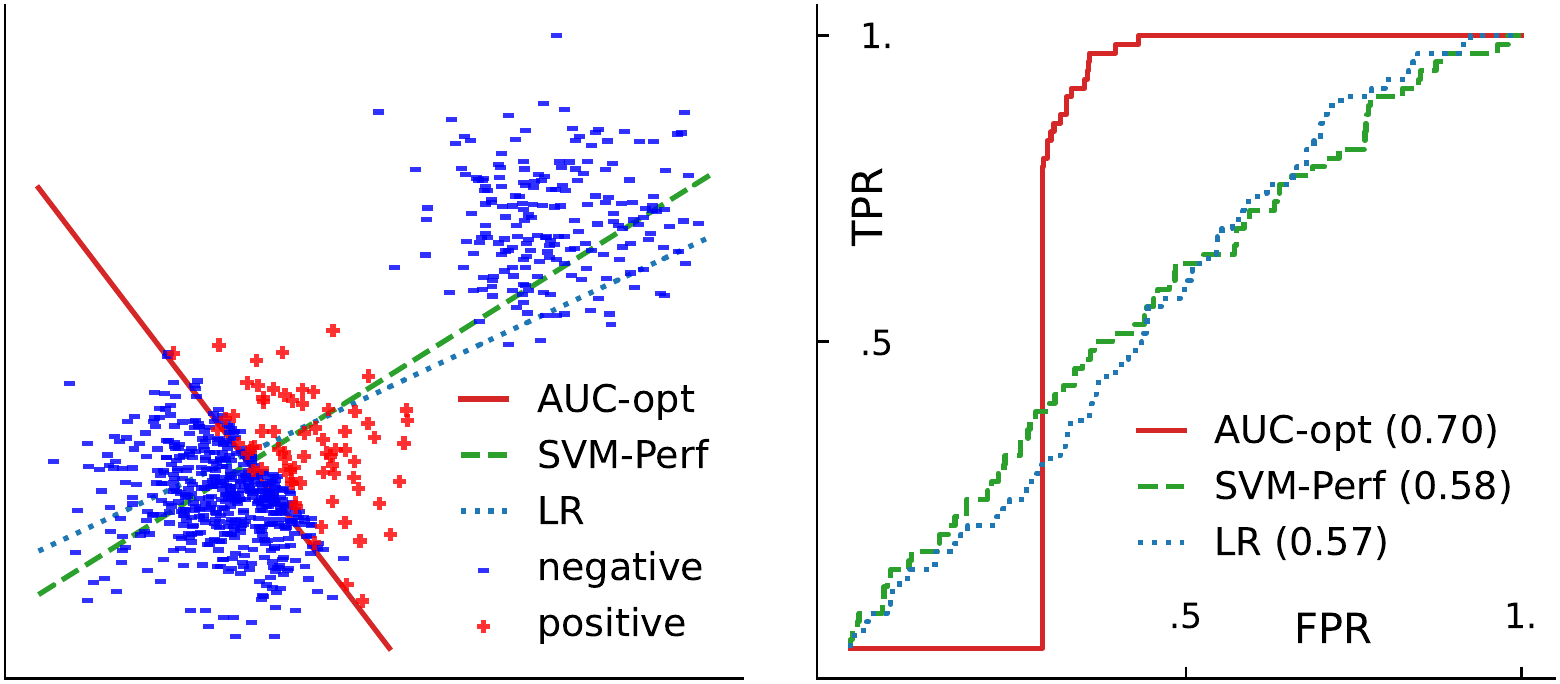}
\caption{The popular AUC classifier SVM-Perf \citep{joachims2005support} fails to find a decent AUC separator on an adversarial example (left), performing similar to Logistic Regression (LR). The corresponding ROC curves and AUC scores (right) for these and our AUC-opt, which beats SVM-Perf and LR by a large margin.\vspace{-6mm}}
\label{fig:motivation-example}
\end{figure}

Fig. \ref{fig:motivation-example} presents a toy example as an illustration where there are significant improvements. To further validate AUC-opt, we conduct experiments on 50 real-world datasets projected onto both $\mathbb{R}^2$ and $\mathbb{R}^3$ by using t-SNE \citep{maaten2008visualizing}. Experiments comparing AUC-opt against seven linear classifiers show that AUC-opt achieves statistically significant improvements on between 17 to 40 of 50 t-SNE datasets at the training stage. To summarize, our main contributions are:

\begin{itemize}
\item  For the first time, we prove the linear AUC optimization problem is NP-complete when $n$ and $d$ are arbitrary but polynomial-time solvable for fixed $d$. A key to our proof is a reduction of the open hemisphere problem. Although NP-hard of the problem was often reported, we have not identified proof in the literature.

\item We then present AUC-opt, to find the provably optimal AUC linear classifier in $\mathbb{R}^2$. It runs in $\mathcal{O}(n_+ n_- \log (n_+ n_-))$, which is optimal under the \textit{algebraic computation tree model} \citep{ben1983lower}. It can be naturally extended to $\mathbb{R}^d$ in $\mathcal{O}( (n_+ n_-)^{d-1} \log (n_+ n_-))$ by decomposing original problem into same subproblems of lower dimensional spaces and calling AUC-opt recursively.

\item Experiments comparing AUC-opt against seven other classification methods show that AUC-opt achieves significant improvements on between 17 to 40 in $\mathbb{R}^2$ and on between 4 to 42 in $\mathbb{R}^3$ of 50 t-SNE training datasets. But generally, the gain proves insignificant on most testing datasets compared to the best standard classifiers. Empirical results suggest that approximate AUC classifiers have space to improve.

\item Similarly, empirical findings on nonlinear classifiers further suggest that the partial loss of significance of approximate AUC optimizers may be due to imperfect approximation, thus having space to improve current approximate algorithms.
\end{itemize}

\subsection{Related Work}

\noindent\textbf{AUC optimization.} In the seminal work, \citet{cortes2004auc} shows that AUC is monotonically increasing with respect to the accuracy and is equal to the accuracy when $n_+ = n_-$. Yet, because the variance is not zero, optimizing directly for AUC may still yield better AUC values than that of standard classifiers. Many AUC optimization methods have been proposed over past years~\citep{yan2003optimizing,herschtal2004optimising,brefeld2005auc,joachims2005support,rakotomamonjy2004optimizing,calders2007efficient,ying2016stochastic,yang2022auc}. These approaches all focus on approximation due to the nonconvex and nondifferentiable of the AUC objective. To avoid this,  \citet{yan2003optimizing} propose to replace the 0-1 objective by a sum of differentiable sigmoid so that a gradient descent-based method can be applied. \citet{joachims2005support} relaxes the problem to a convex one so that SVM can be used (see also \citep{rakotomamonjy2004optimizing,brefeld2005auc}).  A study of \citet{rakotomamonjy2004optimizing} indicates that optimizing SVM objective is also attend to optimize AUC~\citep{rakotomamonjy2004optimizing,steck2007hinge}. More recently, methods for optimizing AUC are studied under the online learning setting \citep{zhao2011online,gao2013one,ying2016stochastic,liu2018fast}. However, performance gains are insignificant found in these studies, and there is a lack of comparison between AUC optimizers and standard methods. Different from these previous works, our goal is to optimize AUC score without approximation.

\noindent\textbf{Bipartite ranking.} The AUC optimization is closely related with the bipartite ranking problem \citep{freund2003efficient,le2010optimization,qin2010general,rudin2018direct} where minimizing the pairwise misranking error is equivalent to maximize the AUC score. For example, RankBoost \citep{freund2003efficient}, a popular ranking algorithm, implicitly optimizes AUC as proved in \cite{cortes2004auc}. \citet{kotlowski2011bipartite} consider maximizing AUC as a minimization of the rank loss. Recently, \citet{rudin2018direct} propose to directly optimize rank statistics by using mixed-integer programming. More works can be found in \citet{krishna2016aditya} and references therein.

\noindent\textbf{Computational complexity results.} Although NP-hard of the problem is often cited as folklore \citep{gao2013one,gao2015consistency,gultekin2020mba}, we have not identified proof in the literature.  Several NP-hardness results have previously been shown for both classifications of 0-1 loss and ranking \citep{feldman2012agnostic,ben2003difficulty,order1999}. \citet{order1999} show that finding the optimal ranking is NP-complete. Although \citet{joachims2005support} shows that the AUC optimization can be reformulated as a classification problem, a proof of NP-hardness from it does not seem to follow naturally.  Instead, we prove the NP-hardness by the reduction from the open hemisphere problem.

\subsection{Paper outline and notations}

The remainder of this paper is organized as follows. We first present preliminaries in \S\ref{section:problem-formulation}. The proof of NP-hardness of AUC optimization is given in \S\ref{section:np-complete}. AUC-opt and its generalization to high-dimensional space are given in \S\ref{section:algorithms}. We empirically evaluate AUC-opt and then make our conclusion in \S\ref{section:experiments} and \S\ref{section:conclusion}, respectively. Throughout this paper, we restrict our attention to optimize AUC in a linear hypothesis class $\mathcal{H}$ i.e. $\bm f \in \mathcal{H}:=\{\bm w: \bm w \in \mathbb{R}^d \}$. Given a set of $n$ training examples $\mathcal{D}:=\{(\bm x_i, y_i) : i \in \left\{1,2,\ldots,n\right\}\}$ where $\bm x_i \in \mathbb{R}^d$ and $y_i \in \{\pm 1\}$, we rewrite $\mathcal{D}$ as a union of $\mathcal{D}_+$ and $\mathcal{D}_-$ where $\mathcal{D}_+$ is the set of positive samples written as $\{(\bm x_1^+, y_1^+),\ldots,(\bm x_{n_+}^+, y_{n_+}^+)\}$ and $\mathcal{D}_-$ is the set of negative samples written as $\{(\bm x_1^-, y_1^-),\ldots,(\bm x_{n_-}^-, y_{n_-}^-)\}$, respectively. Clearly, $n = n_+ + n_-$ and $\mathcal{D} = \mathcal{D}_+ \cup \mathcal{D}_-$. $x_{i j}$ is denoted as $j$-th entry of vector $\bm x_i$, i.e. $\bm x_i = [x_{i1},x_{i1},\ldots,x_{id}]^\top$.

\section{Preliminaries}
\label{section:problem-formulation}

We first review the definition of AUC statistic and give the problem formulation under linear hypothesis $\mathcal{H}$. We discuss that the problem is efficiently solvable when $\mathcal{D}$ is separable.

\begin{definition}[AUC Statistic \citep{clemenccon2008ranking}]
Let $(X,Y)$ and $(X^\prime, Y^\prime)$ be two pairs of random variables in $\mathbb{R}^d\times \{\pm 1\}$ following the same unknown distribution. Denote the probability of an event $A$ condition on $\{Y=1,Y^\prime=-1\}$ as $\mathbb{P}\{A | Y=1,Y^\prime=-1\}$. Given a score function $\bm f: \mathbb{R}^d \mapsto \mathbb{R}$, the AUC statistic is 
\begin{equation}
\operatorname{AUC}(\bm f) := \mathbb{P}\{ \bm f(X) > \bm f(X^\prime) | Y=1,Y^\prime=-1\}. \nonumber
\end{equation}
\end{definition}
Statistically, $\operatorname{AUC}(\bm f)$ is the probability that a randomly chosen positive sample is ranked by $\bm f$ higher than a randomly chosen negative one (tie-breaks lexicographically). It is equivalent to the Wilcoxon statistic \citep{hanley1982meaning}. Given $\mathcal{D}$, our linear AUC optimization is empirically defined as the following.

\begin{problem}[Linear AUC Optimization (LAO)]
\label{problem:lao-problem}
Given the dataset $\mathcal{D}$ and the hypothesis class $\mathcal{H}:=\{\bm w: \bm w \in \mathbb{R}^d\}$, the LAO problem is to find a $\bm w\in \mathcal{H}$ such that the empirical AUC score is maximized, that is
\begin{small}
\begin{equation}
(\textit{LAO}) \quad \bm w^\star \in \argmax_{\bm w \in \mathcal{H}} \sum_{i=1}^{n_+} \sum_{j=1}^{n_-} \frac{{\bm 1} [\bm w^\top \bm x_i^{+}> \bm w^\top \bm x_j^{-} ]}{n_+ n_-}, \label{equ:linear-auc-posi}
\end{equation}
\end{small}
where the indicator ${\bm 1}[A]=1$ if $A$ is true 0 otherwise.
\end{problem}

Due to the non-convexity of 0-1 loss in LAO, directly optimizing (\ref{equ:linear-auc-posi}) is challenging. Notice that $\bm w^\star$ is not unique as $p + \alpha \bm w^\star$ with $\alpha > 0$ and $p \in \mathbb{R}$ is always an optimizer. Although (\ref{equ:linear-auc-posi}) is hard to optimize, if $\mathcal{D}$ is \textit{linearly separable}, that is, there exists a $\bm w$ such that for any $(\bm x_{i}^+,y_{i}^+) \in \mathcal{D}_+$, $\left\langle \bm w, \bm x_{i}^+ \right\rangle \geq 0$ and for any $(\bm x_{j}^-,y_{j}^-) \in \mathcal{D}_-$, $\left\langle\bm w, \bm x_{j}^- \right\rangle < 0$, then one can always find a  $\bm w$ such that $\operatorname{AUC}(\bm w)=1$ in polynomial-time \citep{elizondo2006linear,linear2007} by using Perceptron \citep{freund1999large} or linear programming techniques. For example, the worst time complexity of an iterative reduction algorithm is $\mathcal{O}(n r^3)$ where $r \leq \min(n,d+1)$ \citep{elizondo2006linear}. In the rest of this paper, we assume $\mathcal{D}$ is not linearly separable.

Although previous studies have claimed the NP-hardness of LAO \citep{gao2013one,gao2015consistency,gultekin2020mba}, no previous literature proves it even for the linear classifier case. It motivates us to prove the NP-hardness under the linear hypothesis.

\section{NP-hardness of LAO}
\label{section:np-complete}

This section proves the LAO problem under the linear hypothesis is NP-complete if $n$ and $d$ are not fixed but polynomial-time solvable when $d$ is fixed. We first introduce the open hemisphere problem and then prove the NP-complete by a reduction from it.

\begin{definition}[Open hemisphere]
Given the unit sphere $\mathcal{S}^{d-1}:=\{\bm s \in \mathbb{R}^d: \|\bm s\|_2 = 1\}$, the open hemisphere of $\bm w$ is defined as a set $\{\bm s \in \mathcal{S}^{d-1}: \left\langle \bm w, \bm s \right\rangle > 0 \}$.
\end{definition}

\begin{problem}[Open hemisphere problem~\citep{johnson1978densest}]
Let $\mathcal{K}:=\{\bm s_1,\bm s_2,\ldots, \bm s_t\}$ be a subset of $\mathbb{Q}^d \cap \mathcal{S}^{d-1}$ where $\mathbb{Q}$ is the set of rationals. The open hemisphere problem is to find an open hemisphere such that it contains a largest subset of $\mathcal{K}$, that is,
\begin{equation}
\argmax_{\bm w \in \mathbb{R}^d} |\{\bm s_i \in \mathcal{K}: \left\langle \bm w,\bm s_i \right\rangle > 0 \}|. \nonumber
\end{equation}
\end{problem}
To ease our analysis, we formulate the open hemisphere problem as a feasibility problem. Given positive integers $d$, $m$ and a set $\mathcal{K}$, does there exist a hyperplane $\bm w$ such that at least $m$ inequalities are satisfied, that is,
\begin{equation}
(\textit{OH})\qquad |\{\bm s_i \in \mathcal{K}: \left\langle \bm w,\bm s_i \right\rangle > 0 \}| \geq m \ ? \label{equ:open-hemisphere}
\end{equation}

\begin{lemma}[NP-complete of OH~\citep{johnson1978densest}]
\label{lemma:johnson-np-hard}
Given positive $d$, $n$, $m$, and $\mathcal{K}\subseteq \mathbb{Q}^d\cap \mathcal{S}^{d-1}$, the feasibility of OH problem defined in (\ref{equ:open-hemisphere}) is NP-complete when both $n$ and $d$ are not fixed.
\end{lemma}

The above lemma shows the fact that OH is NP-complete. Based on this lemma, we have the following main theorem.

\begin{theorem}[NP-complete of LAO]
Consider the linear AUC optimization problem defined in Problem \ref{problem:lao-problem}, if $\mathcal{D}$ is not linear separable, LAO is NP-complete when both $n$ and $d$ are arbitrary.\label{thm:np-complete-lao}
\end{theorem}

Before proving Thm. \ref{thm:np-complete-lao}, we first reformulate the LAO problem as a feasibility problem and then show this feasibility problem is NP-complete.

\begin{problem}[The feasibility problem of LAO]
\label{problem:flao}
\textit{Given a finite dataset $\mathcal{D}$, a set of linear classifiers $\mathcal{H}$, and positive integers $d, t$ as an input, the feasibility problem of LAO is to ask, does there exist $ \bm w \in \mathcal{H}$ such that}
\begin{equation}
\sum_{i=1}^{n_+} \sum_{j=1}^{n_-} {\bm 1}[\bm w^\top \bm x_{i}^+ > \bm w^\top \bm x_{j}^-] \geq t \ ?
\label{inequ:feasbility-lao}
\end{equation}
\end{problem}

\textit{Proof Sketch.}\footnote{A detailed proof is in the supplementary.} Without loss of generality, let us assume the problem in $\mathbb{Q}^d$. To prove the NP-complete, we only need to show that the feasibility of LAO defined in Problem \ref{problem:flao} is both in NP and is NP-hard by a reduction from OH. First of all, Problem \ref{problem:flao} is in NP. Given any $\bm w \in \mathcal{H}$, one can find a polynomial-time verifier such that it finishes in $\mathcal{O}(n_p n_q d)$ time, and each certificate has a polynomial length for the input.

To show Problem \ref{problem:flao} is NP-hard, given any instance of the OH problem, the goal is to prove that an instance of (\ref{problem:flao}) can solve it. To do this, we construct the training dataset $\mathcal{D}$ so that an instance of Problem \ref{problem:flao} can be defined. Notice that one can rewrite vectors in $\mathcal{K}$ and construct new vectors $\bm x_{i}^+$ and $\bm x_{1}^-$ as the following
\begin{align}
\bm s_1 = \underbrace{(\bm s_1 + \bm x_{1}^-)}_{\bm x_{1}^+} - \bm x_{1}^-, \ldots, \bm s_t = \underbrace{(\bm s_t + \bm x_{1}^-)}_{\bm x_{t}^+} - \bm x_{1}^-. \label{equ:baojian-equ}
\end{align}

\setlength{\columnsep}{8pt}%
\begin{wrapfigure}[13]{L}{2.5cm}
\includegraphics[width=.9\linewidth]{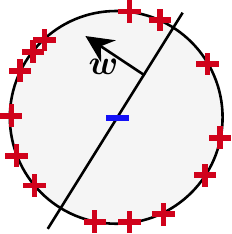}
\includegraphics[width=1.\linewidth]{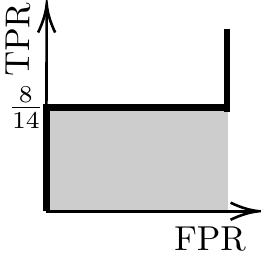}%
\vspace{-15mm}
\label{fig:illustration-lao-oh}
\end{wrapfigure}%

The set of training labels is constructed such that $y_{1}^+, y_{2}^+, \ldots, y_{t}^+$ are all ones and $y_{1}^-=-1$. Combining it with equations in (\ref{equ:baojian-equ}) provides a dataset $\mathcal{D} =\{(\bm x_{1}^+, y_{1}^+),\ldots, (\bm x_{t}^+, y_{t}^+), (\bm x_{1}^-, y_{1}^-)\}$. The left two figures illustrate this reduction where the top figure is a sphere, and each positive sign represents $\bm s_i$ while the negative sign is $\bm x_{1}^- = \bm 0$. The sphere contains 14 points, which correspond to 14 inequalities of the left-hand side of (\ref{inequ:feasbility-lao}). The normal $\bm w$ defines a hyperplane which corresponds to $t=8$ of the right-hand side of (\ref{inequ:feasbility-lao}). The bottom figure shows an AUC curve where TPR and FPR are true positive rate and false positive rate, respectively. It indicates that there exists $\bm w$ such that the number of inequalities that can be satisfied is at least $m$. This transformation and checking procedure can be done in polynomial-time. Therefore, any answer of the instance of LAO is affirmative if and only if the instance of OH is affirmative; hence, the problem is NP-hard. \hfill\qedsymbol

\section{Proposed methods for LAO}
\label{section:algorithms}
In this section, we first present a trivial method in $\mathbb{R}^2$ and then propose AUC-opt inspiring from \textit{topological sweeping}. We then extend AUC-opt to $\mathbb{R}^d$ by projecting high-dimensional problems into low-dimensional ones.

\subsection{A trivial method} 

Notice that every pair of training samples defines a supporting line that separates the rest training samples, and the number of these interesting lines are at most $\mathcal{O}(n^2)$. Given any two sample $\bm x_{i}:=[x_{i1},x_{i2}]^\top$ and $\bm x_{j}:=[x_{j1},x_{j2}]^\top$, the slope of the line defined by $\bm x_i$ and $\bm x_j$ is $m = -(x_{i2} - x_{j2})/(x_{i1} - x_{j1})$. Hence, one can find an algorithm runs in $\mathcal{O}(n^3 \log n)$ and takes $\mathcal{O}(n^2)$ space by using the following two steps: 1) identify all $n(n-1)$ slopes $m$; and 2) for each slope  $m$, let $\bm w := [m, 1]^\top$ and then calculate $\operatorname{AUC}(\bm w)$ by using an $\mathcal{O}(n \log n)$ algorithm \citep{fawcett2006introduction}. $\bm w^\star$ is the one that gives the highest $\operatorname{AUC}(\bm w)$ value.  

\subsection{AUC-opt}

However, the above trivial method can be significantly improved. There are two successive algorithmic improvements needed. The first is that the number of interesting slopes is at most $\mathcal{O}(n_+ n_-)$ by noticing that training pairs with the same labels are not interesting. The second improvement is inspired by the \textit{topological sweep}~\citep{edelsbrunner1989topologically} that we can save $n$ running time by sorting all slopes once; that is, \textit{we only need to calculate the AUC score once for the minimal slope and then sweep over all the rest of the slopes in ascending order.}

\begin{algorithm}[H]
\small
\caption{$[\operatorname{AUC}_{\operatorname{opt}},\bm w]=$AUC-opt($\mathcal{D}$)}
\begin{algorithmic}[1]
\STATE $S =\{\}$  \hfill $\triangleright$ Initialize a slope set
\FOR{$(i, j) \in \{1,\ldots,n_+\} \times \{1,\ldots,n_-\}$} 
\STATE $m = -(x_{i2}^+ - x_{j2}^-)/(x_{i1}^+ - x_{j1}^-)$ \hfill $\triangleright$ Calculate a slope
\STATE $S = S\cup \{(m, i, j)\}$
\ENDFOR
\STATE Let $\epsilon < \min_{(i,j) \ne (i^\prime,j^\prime)} \left| \frac{ x_{i2} -  x_{j2}}{ x_{i1} - x_{j1}} -\frac{ x_{i^\prime2} - x_{j^\prime2}}{ x_{i^\prime1} - x_{j^\prime1}} \right|$ be a positive constant smaller than the minimal difference between any two slopes.
\STATE Let $\min(S) = \min \{m | (m,i,j) \in S\}$
\STATE $\bm w = [\min(S)-1,1]^\top$
\STATE $\operatorname{AUC}_{\operatorname{cur}} = \operatorname{AUC}(\bm w)$ \hfill $\triangleright$ Algo. 1 in \citet{fawcett2006introduction}
\STATE $\operatorname{AUC}_{\operatorname{opt}} = \operatorname{AUC}_{\operatorname{cur}}$
\FOR{$(m,i,j) \in \text{sort}(S)$} 
\STATE \quad $c_1 = 0, c_2 = 0, \bm u = [m - \epsilon,1]^\top, \bm v = [m + \epsilon,1]^\top$
\STATE \quad \textbf{if} $\bm u ^\top\bm x_{i}^+  > \bm u^\top \bm x_{j}^-$ \textbf{then} \ $c_1 = 1$
\STATE \quad \textbf{if} $\bm v^\top\bm x_{i}^+  > \bm v^\top\bm x_{j}^- $ \textbf{then} \ $c_2 = 1$
\STATE \quad $\operatorname{AUC}_{\operatorname{cur}} = \operatorname{AUC}_{\operatorname{cur}} + (c_2 - c_1) /(n_+ n_-)$
\IF{$\operatorname{AUC}_{\operatorname{opt}} < \operatorname{AUC}_{\operatorname{cur}}$}
\STATE $\operatorname{AUC}_{\operatorname{opt}} = \operatorname{AUC}_{\operatorname{cur}}$,\ $\bm w = [m + \epsilon,1]^\top$
\ENDIF
\IF{$\operatorname{AUC}_{\operatorname{opt}} < 1 - \operatorname{AUC}_{\operatorname{cur}}$}
\STATE $\operatorname{AUC}_{\operatorname{opt}} = 1 - \operatorname{AUC}_{\operatorname{cur}}$,\ $\bm w = [-m - \epsilon,-1]^\top$
\ENDIF
\ENDFOR
\end{algorithmic}
\label{algo:auc-opt}
\end{algorithm}

\begin{wrapfigure}{r}{3cm}
\vspace{-3mm}
\includegraphics[width=.9\linewidth]{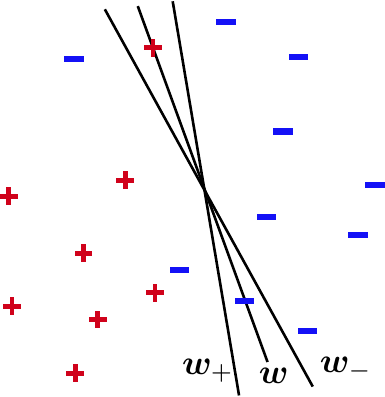}
\vspace{-3mm}
\end{wrapfigure}

The description of AUC-opt is presented in Algo. \ref{algo:auc-opt}. There are three main steps: 1) to obtain all possible slopes and store into $S$ (L1 to L5); 2) to calculate the AUC score of the minimal slope (L7 to L9); and 3) to update $\bm w$ and its AUC score (L10 to L21). The critical part is the update rules of sweeping $\bm w$. More specifically, at each iteration (from L10 to L21), the change of AUC score is from $\{0,\pm 1/(n_+ n_-)\}$ by ``sweeping'' $\bm w$ (L10 - L13). We illustrate this procedure on the upper right figure where red plus signs are positives while blue bar signs are negatives. Let $\bm w = [s,1]^\top, \bm w_- = [s-\epsilon,1]^\top$, and $\bm w_+ = [s+\epsilon,1]^\top$. The increase of slope from $\bm w_-$ to $\bm w_+$ only changes the ordering of two samples which corresponds to a magnitude change of $|1/(n_+ n_-)|$ for $\operatorname{AUC}(\bm w)$. The following theorem shows $\bm w$ returned by Algo. \ref{algo:auc-opt} is optimal.

\begin{theorem}
Given the dataset $\mathcal{D} := \{(\bm x_i, y_i): i \in \left\{1,2,\ldots,n\right\}\}$ where $\bm x_i \in \mathbb{R}^2$ and $y_i \in \{\pm 1\}$ and a collection of linear separators $\mathcal{H}:=\{\bm w: \bm w\in\mathbb{R}^2\}$. The proposed $\operatorname{AUC}$-$\operatorname{opt}$, solves the LAO problem~(\ref{equ:linear-auc-posi}) in $\mathcal{O}(n_+ n_- \log (n_+ n_-))$. This time complexity is tight under the algebraic computation tree model. \label{thm:theorem-run-time}
\end{theorem}
\begin{proof}
We first show that $\bm w$ returned by AUC-opt is optimal. The key of our proof is to show that Algo. \ref{algo:auc-opt} iterates all possible AUC scores given by noticing that all slopes of lines between two consecutive slopes give the same AUC score.

Let $\bm w = [w_1, w_2]^\top$ be any line in $\mathbb{R}^2$. We assume that $\bm w$ is a normal vector of two training samples $\bm x_i^+, \bm x_j^-$, that is, $\bm w$ is given by $\langle \bm w, \bm x_{i}^+ - \bm x_{j}^- \rangle = 0$. The slopes of these normal vectors on $\mathbb{R}^2$ can be calculated. Let the collection of all such slopes be $\mathcal{S} := \left\{-(x_{i2}^+ - x_{j2}^-)/(x_{i1}^+ - x_{j1}^-) : \bm x_{i}^+, \bm x_{j}^- \in \mathcal{D}  \right\}$. Sort the collection of slopes $\mathcal{S}$ and $\mathbb{R}^2$ has been partitioned into $n_+ n_- + 1$ parts. 

We consider two consecutive slopes in the sorted $\mathcal{S}$ more carefully. Let us denote two consecutive sorted slopes as $s_1$ and $s_2$ which associated with two pairs $(\bm x_{i}^+,\bm x_{j}^-)$ and $(\bm x_{{i^\prime}}^+,\bm x_{{j^\prime}}^-)$, respectively. We only need to show that $\forall s \in (s_1,s_2), [s,1]^\top$ has identical AUC score. To do so, all we need to do is to show that given any $\bm x_i^+, \bm x_j^-$, $\bm w$ scores $\bm x_i^+, \bm x_j^-$ as a same ordering. In other words, given any $\bm w: = \{[s,1]^\top: s \in (s_1,s_2)\}$, the quantity $\bm x_i^\top \bm w - \bm x_j^\top \bm w$ always has same sign. In the rest proof, we show this by carefully constructing a quantity function $s(\lambda)$ as the following
\begin{equation}
s(\lambda) := -\frac{x_{i2}^+ - x_{j2}^-}{x_{i1}^+ - x_{j1}^-} \lambda - \frac{x_{{i2}^\prime}^+ - x_{{j2}^\prime}^-}{x_{{i1}^\prime}^+ - x_{{j1}^\prime}^-} (1 - \lambda), \nonumber
\end{equation}
where $\lambda \in (0,1)$. 

We need to study the monotonicity of $s(\lambda)$, we define another function  $h(\lambda) := s(\lambda)\left(x_{i1} - x_{j1}\right) + x_{i2} - x_{j2}$.
Clearly, $h(\lambda)$ is non-decreasing function by noticing that $h^\prime(\lambda) = s_2 - s_1 \geq 0$. We just need to show $h(\lambda )$ never vanishes at $\lambda \in (0,1)$. Assume that we have $h(\lambda) = 0$, then we have $s(\lambda) = - (x_{i2} - x_{j2})/(x_{i1} - x_{j1})$. It makes a contradiction since there is no existing slope between $s_1$ and $s_2$. Similarly, one can show that for any two consecutive slopes $s_1,s_2$, $\forall s \in (s_1,s_2), \bm w:=[-s,-1]^\top$, also defines the same AUC score. Since Algo. \ref{algo:auc-opt} iterates all such lines, the best AUC score of $\bm w$ returned by AUC-opt is indeed optimal.

AUC-opt finishes in $\mathcal{O}(n_+ n_-\log (n_+ n_-))$ since the time complexity is dominated by sorting all $n_+ n_-$ slopes (L10). The tightness of time complexity follows by Lemma 3.6.16 of \citet{lee1984computational}.
\end{proof}

\subsection{Generalization to $\mathbb{R}^d$} 

When problem dimension $d \geq 3$, inspired from \citet{johnson1978densest}, the general idea of solving high-dimensional LAO problem is that one can decompose the original $d$ dimensional problem into several $d-1$ subproblems. Notice that each hyperplane $H(\bm u)$ uniquely defines an interesting subspace, and there are at most $n_+ n_-$ such subspaces. We project points onto $H(\bm u)$ and then solve problem in $d-1$ dimensional subspace (changing the number of coordinates from $d$ to $d-1$), recursively. Specifically, let the projection be defined as $\bm P(\bm x) := \bm x - (\bm x^\top \cdot \bm u / \| \bm u\|^2) \cdot \bm u$ where $\bm u$ is the normal vector of $H(\bm u)$. The critical property of $\bm P$ is that for any $\bm x \in H(\bm u)$, $\bm x^\top \bm P(\bm x_i^+ - \bm x_j^-) = \bm x^\top (\bm x_i^+ - \bm x_j^-)$.  Therefore, the inner products of training samples with $\bm w$ are the same as those of projected training samples. 

Due to this inevitable recursion, the number of interesting hyperplanes exponentially depends upon $d$; hence time complexity of AUC-opt in $\mathbb{R}^d$ is $\mathcal{O}\left((n_p n_q)^{d-1}\log (n_p n_q)\right)$. We summarize this recursive procedure in Algo. \ref{algo:auc-opt-3d}. Due to the space limitation, detailed algorithm description is in the supplementary.

\begin{algorithm}[H]
\small
\caption{$[\operatorname{AUC}_{\operatorname{opt}},\bm w]=$AUC-opt($\mathcal{D},d$)}
\begin{algorithmic}[1]
\STATE $\mathcal{K} = \{(\bm x_i^+ - \bm x_j^-): i \in \{1,\ldots,n_+\},j \in \{1,\ldots,n_-\}\}$
\IF {d=2}
\STATE  \textbf{return} $\operatorname{AUC}_{cur}, \bm w^\prime = \operatorname{AUC-opt}(\mathcal{D})$ \hfill $\triangleright$ call Algo. \ref{algo:auc-opt}
\ENDIF
\FOR{$\bm u \in \mathcal{K}$}
\STATE $\mathcal{P} = \left\{\left(\bm x - \frac{\bm x^\top \cdot \bm u}{\| \bm u\|^2} \cdot \bm u, y\right): (\bm x, y) \in \mathcal{D}\right\}$ \hfill $\triangleright$ project points of $\mathcal{D}$ onto the hyperplane defined by $\bm u$.
\STATE $\mathcal{P}^\prime = \operatorname{change\_coordinates}(\mathcal{P})$ \hfill $\triangleright$ change coordinates so that points are presented in $d-1$ coordinates.
\STATE $\operatorname{AUC}_{cur}, \bm w^\prime  =\operatorname{AUC-opt}(\mathcal{P}^\prime, d-1)$
\STATE $\bm w = \operatorname{change\_coordinates}(\bm w^\prime)$ \hfill $\triangleright$ change $d-1$ coordinates of $\bm w^\prime$ back to $d$ coordinates.
\IF {$\operatorname{AUC}_{\operatorname{opt}} < \operatorname{AUC}_{\operatorname{cur}}$}
\STATE $\bm w^\star = \bm w$, $\operatorname{AUC}_{\operatorname{opt}} = \operatorname{AUC}_{\operatorname{cur}}$
\ENDIF
\ENDFOR
\end{algorithmic}
\label{algo:auc-opt-3d}
\end{algorithm}

\begin{theorem}
\label{thm:generalization}
Given the dataset $\mathcal{D} := \{(\bm x_i, y_i): i \in \left\{1,2,\ldots,n\right\}\}$ where $\bm x_i \in \mathbb{R}^d$ and $y_i \in \{\pm 1\}$ and a collection of linear separators $\mathcal{H}:=\{\bm w: \bm w\in\mathbb{R}^d\}$. There exists an algorithm solves the LAO problem (\ref{equ:linear-auc-posi}) exactly in $\mathcal{O}\left((n_+ n_-)^{d-1}\log (n_+ n_-)\right)$. 
\end{theorem}

\begin{proof}
The proof is in the supplementary.
\end{proof}

\newcommand*\rot{\rotatebox{60}}
\definecolor{Gray}{gray}{0.9}
\makeatletter
\newcommand{\thickhline}{%
    \noalign {\ifnum 0=`}\fi \hrule height 1pt
    \futurelet \reserved@a \@xhline
}
\newcolumntype{"}{@{\hskip\tabcolsep\vrule width .5pt\hskip\tabcolsep}}
\makeatother
\setlength\tabcolsep{.02pt}
\begin{table*}[h]
\small
\caption{The comparison of AUC scores over 200 trials of 50 t-SNE datasets on both $\mathbb{R}^2$ and $\mathbb{R}^3$. Each cell (I,J) means the number of datasets where method I is significantly better than method J by using t-test with a significance level of 5\%. Numbers in the red region are results of AUC optimizers better than standard classifiers, while numbers in the blue region are the reverse.\vspace{-2mm}}
\centering
\begin{tabular}{P{0.05\textwidth}"p{0.1\textwidth}P{0.052\textwidth}P{0.052\textwidth}P{0.052\textwidth}P{0.052\textwidth}P{0.052\textwidth}P{0.052\textwidth}P{0.052\textwidth}P{0.052\textwidth}"P{0.052\textwidth}P{0.052\textwidth}P{0.052\textwidth}P{0.052\textwidth}P{0.052\textwidth}P{0.052\textwidth}P{0.052\textwidth}P{0.052\textwidth}}\toprule 
&  & \rot{SVM} & \rot{B-SVM} & \rot{LR} & \rot{B-LR} & \rot{\small{SVM-Perf}} & \rot{SPAUC} & \rot{SPAM} & \rot{AUC-opt} & \rot{SVM} & \rot{B-SVM} & \rot{LR} & \rot{B-LR} & \rot{\small{SVM-Perf}} & \rot{SPAUC} & \rot{SPAM} & \rot{AUC-opt} \\\hline 
& &\multicolumn{8}{c}{Significance t-test ($\alpha=0.05$) on \textit{training}} & \multicolumn{8}{c}{Significance t-test ($\alpha=0.05$) on \textit{testing}} \\
& SVM & - & 0 & 1 & 1 & \cellcolor{blue!30!white}3 & \cellcolor{blue!30!white}0 & \cellcolor{blue!30!white}1 & \cellcolor{blue!30!white}0 & - & 0 & 0 & 0 & \cellcolor{blue!30!white}3 & \cellcolor{blue!30!white}0 & \cellcolor{blue!30!white}1 & \cellcolor{blue!30!white}1 \\
& B-SVM & 30 & - & 7 & 4 & \cellcolor{blue!30!white}27 & \cellcolor{blue!30!white}3 & \cellcolor{blue!30!white}3 & \cellcolor{blue!30!white}0 & 31 & - & 7 & 4 & \cellcolor{blue!30!white}23 & \cellcolor{blue!30!white}3 & \cellcolor{blue!30!white}3 & \cellcolor{blue!30!white}1 \\ 
& LR & 31 & 11 & - & 2 & \cellcolor{blue!30!white}27 & \cellcolor{blue!30!white}2 & \cellcolor{blue!30!white}5 & \cellcolor{blue!30!white}0 & 32 & 14 & - & 2 & \cellcolor{blue!30!white}26 & \cellcolor{blue!30!white}2 & \cellcolor{blue!30!white}5 & \cellcolor{blue!30!white}2 \\
& B-LR & 32 & 15 & 6 & - & \cellcolor{blue!30!white}30 & \cellcolor{blue!30!white}2 & \cellcolor{blue!30!white}4 & \cellcolor{blue!30!white}0 & 33 & 16 & 6 & - & \cellcolor{blue!30!white}29 & \cellcolor{blue!30!white}3 & \cellcolor{blue!30!white}4 & \cellcolor{blue!30!white}2 \\
$d=2$ & SVM-Perf & \cellcolor{red!30!white}16 & \cellcolor{red!30!white}2 & \cellcolor{red!30!white}0 & \cellcolor{red!30!white}0 & - & 0 & 1 & 0 & \cellcolor{red!30!white}16 & \cellcolor{red!30!white}2 & \cellcolor{red!30!white}0 & \cellcolor{red!30!white}1 & - & 1 & 2 & 2 \\
& SPAUC & \cellcolor{red!30!white}31 & \cellcolor{red!30!white}14 & \cellcolor{red!30!white}8 & \cellcolor{red!30!white}2 & 28 & - & 2 & 0 & \cellcolor{red!30!white}29 & \cellcolor{red!30!white}12 & \cellcolor{red!30!white}7 & \cellcolor{red!30!white}3 & 26 & - & 4 & 2 \\
& SPAM & \cellcolor{red!30!white}33 & \cellcolor{red!30!white}10 & \cellcolor{red!30!white}9 & \cellcolor{red!30!white}3 & 29 & 1 & - & 0 & \cellcolor{red!30!white}32 & \cellcolor{red!30!white}11 & \cellcolor{red!30!white}6 & \cellcolor{red!30!white}3 & 26 & 1 & - & 2 \\
& AUC-opt & \cellcolor{red!30!white}\textbf{40} & \cellcolor{red!30!white}\textbf{29} & \cellcolor{red!30!white}\textbf{20} & \cellcolor{red!30!white}\textbf{17} & \textbf{40} & \textbf{21} & \textbf{26} & - & \cellcolor{red!30!white}\textbf{34} & \cellcolor{red!30!white}\textbf{18} & \cellcolor{red!30!white}\textbf{15} & \cellcolor{red!30!white}\textbf{11} & \textbf{31} & \textbf{13} & \textbf{17} & - \\\bottomrule
&  SVM &  - & 0 & 0 & 0 & \cellcolor{blue!30!white}0 & \cellcolor{blue!30!white}0 & \cellcolor{blue!30!white}0 & \cellcolor{blue!30!white}0 & - & 0 & 0 & 0 & \cellcolor{blue!30!white}1 & \cellcolor{blue!30!white}1 & \cellcolor{blue!30!white}0 & \cellcolor{blue!30!white}1 \\
 &  B-SVM &  25 & - & 0 & 0 & \cellcolor{blue!30!white}24 & \cellcolor{blue!30!white}0 & \cellcolor{blue!30!white}0 & \cellcolor{blue!30!white}0 & 24 & - & 1 & 1 & \cellcolor{blue!30!white}20 & \cellcolor{blue!30!white}4 & \cellcolor{blue!30!white}3 & \cellcolor{blue!30!white}3 \\
 &  LR &  34 & 6 & - & 0 & \cellcolor{blue!30!white}32 & \cellcolor{blue!30!white}5 & \cellcolor{blue!30!white}7 & \cellcolor{blue!30!white}0 & 38 & 14 & - & 1 & \cellcolor{blue!30!white}34 & \cellcolor{blue!30!white}10 & \cellcolor{blue!30!white}12 & \cellcolor{blue!30!white}\textbf{8} \\
 &  B-LR &  34 & 6 & 0 & - & \cellcolor{blue!30!white}32 & \cellcolor{blue!30!white}7 & \cellcolor{blue!30!white}9 & \cellcolor{blue!30!white}0 & 36 & 12 & 0 & - & \cellcolor{blue!30!white}33 & \cellcolor{blue!30!white}9 & \cellcolor{blue!30!white}11 & \cellcolor{blue!30!white}\textbf{6} \\
 $d=3$ &  SVM-Perf &  \cellcolor{red!30!white}1 & \cellcolor{red!30!white}0 & \cellcolor{red!30!white}0 & \cellcolor{red!30!white}0 & - & 0 & 0 & 0 & \cellcolor{red!30!white}1 & \cellcolor{red!30!white}0 & \cellcolor{red!30!white}0 & \cellcolor{red!30!white}0 & - & 1 & 2 & 3 \\
 &  SPAUC &  \cellcolor{red!30!white}24 & \cellcolor{red!30!white}0 & \cellcolor{red!30!white}0 & \cellcolor{red!30!white}0 & 27 & - & 0 & 0 & \cellcolor{red!30!white}24 & \cellcolor{red!30!white}3 & \cellcolor{red!30!white}0 & \cellcolor{red!30!white}1 & 22 & - & 1 & 3 \\
 &  SPAM &  \cellcolor{red!30!white}26 & \cellcolor{red!30!white}0 & \cellcolor{red!30!white}0 & \cellcolor{red!30!white}0 & 23 & 0 & - & 0 & \cellcolor{red!30!white}25 & \cellcolor{red!30!white}2 & \cellcolor{red!30!white}0 & \cellcolor{red!30!white}0 & 21 & 2 & - & 2 \\
 &  AUC-opt &  \cellcolor{red!30!white}\textbf{42} & \cellcolor{red!30!white}\textbf{31} & \cellcolor{red!30!white}\textbf{6} & \cellcolor{red!30!white}\textbf{4} & \textbf{41} & \textbf{27} & \textbf{33} & - & \cellcolor{red!30!white}\textbf{32} & \cellcolor{red!30!white}\textbf{15} & \cellcolor{red!30!white}2 & \cellcolor{red!30!white}2 & \textbf{32} & \textbf{9} & \textbf{10} & - \\\bottomrule
\end{tabular}
\label{tab:auc-scores}
\vspace{-2mm}
\end{table*}

\section{Experiments}
\label{section:experiments}

We evaluate AUC-opt on both $\mathbb{R}^2$ and $\mathbb{R}^3$ by using t-SNE datasets. To confirm that AUC-opt produces the best possible linear AUC classifiers, we compare it with 7 other classifiers on the binary classification task at the training and testing stage. Furthermore, we compare the approximate AUC methods with other standard methods for nonlinear classifiers. \textit{More results and experimental details, including data collection, baseline description, and parameter tuning, are in the supplementary.}\footnote{Our code can be found in \url{https://github.com/baojian/auc-opt}}

\subsection{Datasets and experimental setup}

\paragraph{Datasets.} We collect 50 real-world datasets where the positive ratio ($n_p/n$) of most datasets are $\leq 0.1$. These highly imbalanced datasets make optimizing AUC problem meaningful. To generate 2 and 3 dimensional samples, we project samples of these datasets onto $\mathbb{R}^2$ and $\mathbb{R}^3$ respectively using t-SNE \citep{maaten2008visualizing} so that class patterns are conserved. Models use projected points as training samples while keep labels unchanged.
 
 \noindent\textbf{Experimental setup.} For each dataset, 50\% samples are for training and the rest for testing. All parameters are tuned by 5-fold cross-validation. Each dataset is randomly shuffled 200 times, and the reported results are averaged on 200 trials. All methods have been tested on servers with Intel(R) Xeon(R) CPU (2.30GHz) 64 cores and 187G memory. For all methods that involve randomness, the random state has been fixed for the purpose of reproducibility.

\noindent\textbf{Baselines.} We consider the following baseline classifiers:  1) Logistic Regression (LR); 2) B-LR. Balanced Logistic Regression (B-LR). We adjust weights of samples inversely proportional to class frequencies so that it can have better performance on imbalanced datasets; 3) Support-vector Machine (SVM); 4) B-SVM. The balanced SVM (B-SVM) uses the same strategy as B-LR; 5) SVM-Perf. The SVM-Perf algorithm is proposed in \citet{joachims2005support} where the goal is to minimize the AUC loss by using SVM-based method; 6) SPAUC. The Stochastic Proximal AUC (SPAUC) maximization algorithm is proposed in \citet{lei2019stochastic}; 7) SPAM. The Stochastic Proximal AUC Maximization (SPAM) algorithm is proposed in \citet{natole2018stochastic}.

\subsection{Results on t-SNE datasets}

\begin{table*}[t]
\small
\caption{The comparison of testing AUC scores of 50 real-world datasets. The settings are the same as in Table~\ref{tab:auc-scores}.}
\centering
\begin{tabular}{p{0.14\textwidth}P{0.055\textwidth}P{0.055\textwidth}P{0.055\textwidth}P{0.055\textwidth}P{0.055\textwidth}P{0.055\textwidth}P{0.055\textwidth}P{0.055\textwidth}P{0.055\textwidth}P{0.055\textwidth}P{0.055\textwidth}P{0.055\textwidth}P{0.055\textwidth}P{0.055\textwidth}P{0.08\textwidth}}\toprule 
& \rot{SVM} & \rot{B-SVM} & \rot{LR} & \rot{B-LR} & \rot{RF} & \rot{B-RF} & \rot{GB} & \rot{SVM-RBF} & \rot{\scriptsize{B-SVM-RBF}} & \rot{RankBoost} & \rot{AdaBoost} & \rot{SPAM} & \rot{SPAUC} & \rot{SVM-Perf} & \rot{\textbf{Average}} \\\hline 
SVM   & -  & 4  & 2  & 3  & 10  & 9  & 15  & 19  & 8  & \cellcolor{blue!30!white}15  & \cellcolor{blue!30!white}17  & \cellcolor{blue!30!white}29  & \cellcolor{blue!30!white}30  & \cellcolor{blue!30!white}3 & 12.62 \\
B-SVM   & 22  & -  & 6  & 4  & 15  & 12  & 18  & 21  & 12  & \cellcolor{blue!30!white}15  & \cellcolor{blue!30!white}19  & \cellcolor{blue!30!white}36  & \cellcolor{blue!30!white}38  & \cellcolor{blue!30!white}20 & 18.31 \\
LR   & 31  & 21  & -  & 5  & 17  & 14  & 19  & 23  & 14  & \cellcolor{blue!30!white}18  & \cellcolor{blue!30!white}20  & \cellcolor{blue!30!white}38  & \cellcolor{blue!30!white}38  & \cellcolor{blue!30!white}29 & 22.08 \\
B-LR   & 31  & 23  & 8  & -  & 17  & 14  & 19  & 23  & 14  & \cellcolor{blue!30!white}19  & \cellcolor{blue!30!white}21  & \cellcolor{blue!30!white}37  & \cellcolor{blue!30!white}37  & \cellcolor{blue!30!white}28 & 22.38 \\
RF   & 34  & 30  & 31  & 31  & -  & 6  & 34  & 29  & 19  & \cellcolor{blue!30!white}34  & \cellcolor{blue!30!white}38  & \cellcolor{blue!30!white}39  & \cellcolor{blue!30!white}37  & \cellcolor{blue!30!white}34 & 30.46 \\
B-RF   & 36  & 33  & 31  & 31  & 19  & -  & 33  & 32  & 19  & \cellcolor{blue!30!white}35  & \cellcolor{blue!30!white}37  & \cellcolor{blue!30!white}38  & \cellcolor{blue!30!white}36  & \cellcolor{blue!30!white}35 & 31.92 \\
GB   & 31  & 29  & 29  & 29  & 12  & 13  & -  & 25  & 17  & \cellcolor{blue!30!white}26  & \cellcolor{blue!30!white}30  & \cellcolor{blue!30!white}36  & \cellcolor{blue!30!white}38  & \cellcolor{blue!30!white}32 & 26.69 \\
SVM-RBF   & 27  & 26  & 26  & 27  & 16  & 15  & 21  & -  & 6  & \cellcolor{blue!30!white}23  & \cellcolor{blue!30!white}23  & \cellcolor{blue!30!white}30  & \cellcolor{blue!30!white}31  & \cellcolor{blue!30!white}27 & 22.92 \\
B-SVM-RBF   & 36  & 32  & 32  & 32  & 24  & 22  & 29  & 20  & -  & \cellcolor{blue!30!white}31  & \cellcolor{blue!30!white}33  & \cellcolor{blue!30!white}39  & \cellcolor{blue!30!white}38  & \cellcolor{blue!30!white}37 & 31.15 \\
RankBoost   & \cellcolor{red!30!white}32  & \cellcolor{red!30!white}31  & \cellcolor{red!30!white}29  & \cellcolor{red!30!white}29  & \cellcolor{red!30!white}9  & \cellcolor{red!30!white}8  & \cellcolor{red!30!white}16  & \cellcolor{red!30!white}24  & \cellcolor{red!30!white}13  & -  & 35  & 38  & 38  & 33 & 25.77 \\
AdaBoost   & \cellcolor{red!30!white}31  & \cellcolor{red!30!white}29  & \cellcolor{red!30!white}26  & \cellcolor{red!30!white}27  & \cellcolor{red!30!white}7  & \cellcolor{red!30!white}8  & \cellcolor{red!30!white}9  & \cellcolor{red!30!white}22  & \cellcolor{red!30!white}12  & 3  & -  & 37  & 35  & 31 & 21.31 \\
SPAM   & \cellcolor{red!30!white}13  & \cellcolor{red!30!white}4  & \cellcolor{red!30!white}2  & \cellcolor{red!30!white}2  & \cellcolor{red!30!white}7  & \cellcolor{red!30!white}7  & \cellcolor{red!30!white}12  & \cellcolor{red!30!white}18  & \cellcolor{red!30!white}7  & 8  & 12  & -  & 27  & 15 & 10.31 \\
SPAUC   & \cellcolor{red!30!white}6  & \cellcolor{red!30!white}3  & \cellcolor{red!30!white}2  & \cellcolor{red!30!white}1  & \cellcolor{red!30!white}9  & \cellcolor{red!30!white}7  & \cellcolor{red!30!white}10  & \cellcolor{red!30!white}17  & \cellcolor{red!30!white}7  & 10  & 10  & 10  & -  & 8 & 7.69 \\
SVM-Perf   & \cellcolor{red!30!white}3  & \cellcolor{red!30!white}10  & \cellcolor{red!30!white}4  & \cellcolor{red!30!white}5  & \cellcolor{red!30!white}9  & \cellcolor{red!30!white}8  & \cellcolor{red!30!white}15  & \cellcolor{red!30!white}16  & \cellcolor{red!30!white}7  & 13  & 16  & 25  & 28  & - & 12.23 \\\bottomrule
\end{tabular}
\label{tab:auc-scores-real}
\vspace{-3mm}
\end{table*}

\begin{figure}[h]
\centering
\begin{subfigure}[b]{.25\textwidth}
\includegraphics[height=3.5cm,width=4cm]{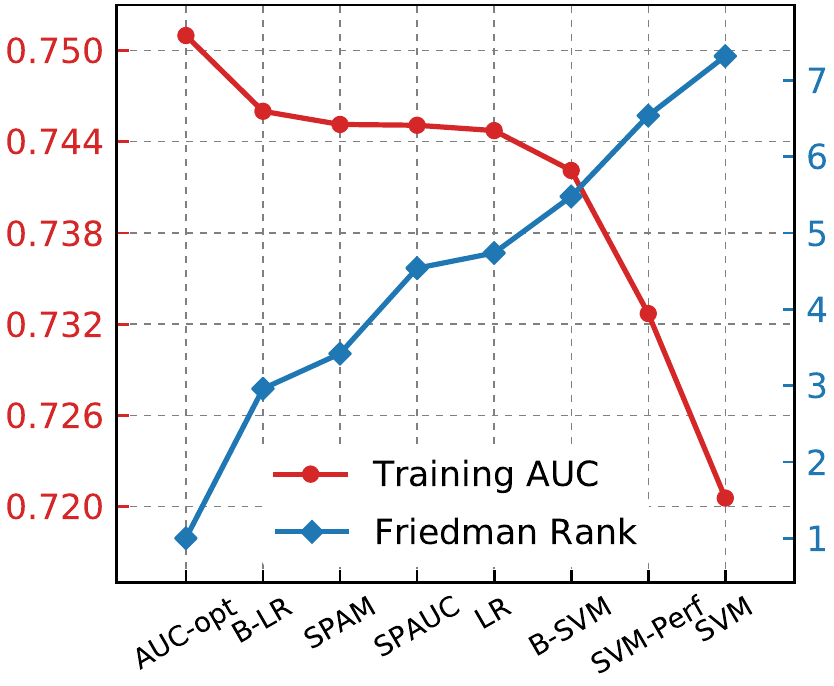}
\label{fig:rank-train-auc}
\caption{Mean of training AUC}
\end{subfigure}%
\begin{subfigure}[b]{.25\textwidth}
\includegraphics[height=3.5cm,width=4cm]{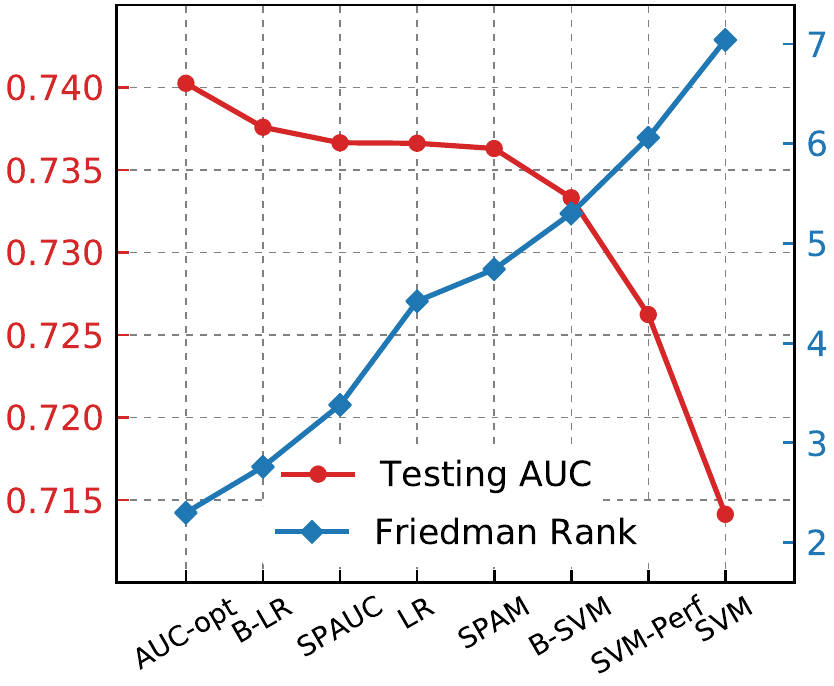}
\label{fig:rank-test-auc}
\caption{Mean of testing AUC}
\end{subfigure}
\begin{subfigure}[b]{.25\textwidth}
\includegraphics[scale=.45]{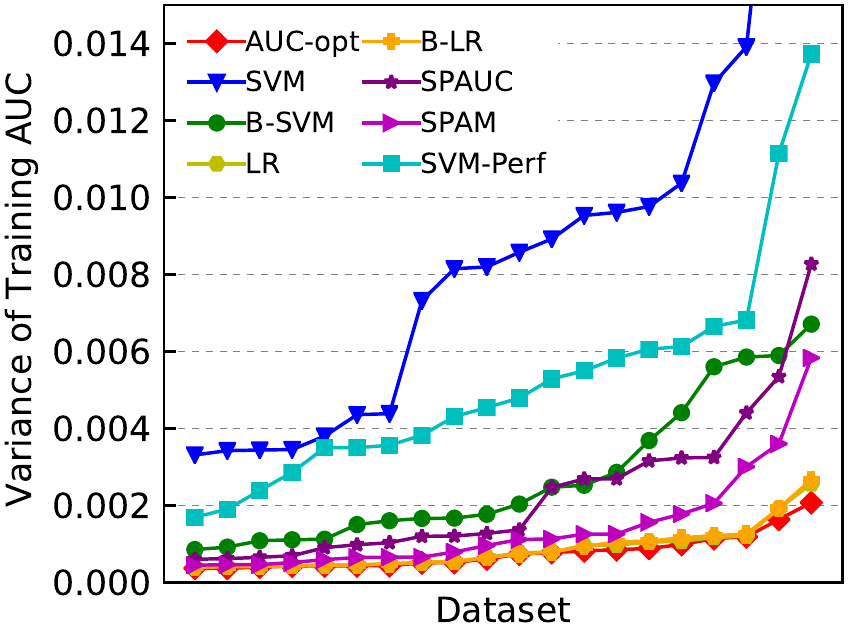}
\label{fig:tr-te-auc-var}
\caption{Variance of training AUC}
\end{subfigure}%
\begin{subfigure}[b]{.25\textwidth}
\includegraphics[scale=.45]{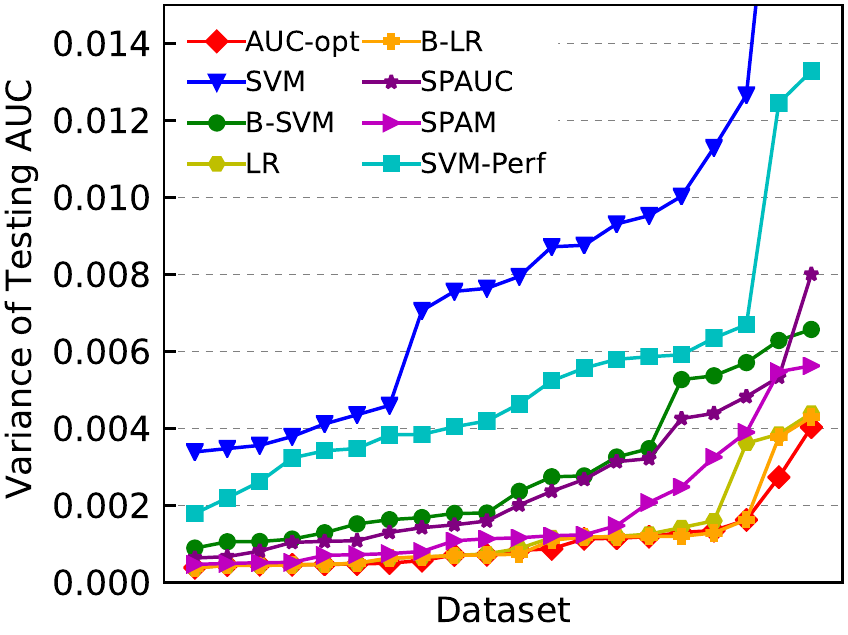}
\caption{Variance of testing AUC}
\end{subfigure}
\caption{Mean and Variance of AUC scores of eight classifiers over 50 t-SNE datasets in $\mathbb{R}^2$.}
\label{fig:rank-test-auc-var}
\vspace{-3mm}
\end{figure}

\noindent\textbf{Comparison of AUC scores.} Table \ref{tab:auc-scores} presents the comparison of AUC scores calculated on both training and testing datasets in $\mathbb{R}^2$ and $\mathbb{R}^3$. To compare AUC scores from different methods, for the method I and J on a specific dataset, we calculate whether I is significantly better than J in a statistical sense. Important observations are: 1) AUC-opt achieves more significant gains over standard classifiers than the approximate AUC optimizers on the training stage. It confirms that gaps between AUC-opt and other approximate AUC optimizers are significant on some datasets; 
2) Compared with SVM and LR, balanced versions of SVM (B-SVM) and LR (B-LR) have better performance. We see that the simple weighting strategy improves the performance by adjusting the weights on training samples;\footnote{The balanced version of a classifier is that a class weight has been added for each sample. Specifically, weight $n/(2*n_p)$ is for each positive sample while weight $n/(2*n_q)$ is for each negative sample.} and 3) Compared with the best standard classifier B-LR, AUC-opt produces significant gains on 17 training datasets while reduced to 11 on testing datasets. It is worse in $\mathbb{R}^3$. This degradation happens mainly because results of B-LR are tuned by adding regularization hence better generalization ability, while AUC-opt does not. With regularization, B-LR has better generalization performance on testing datasets but regularization is not taking into consideration in our method. Under $\mathbb{R}^3$, AUC-opt only beats 4 datasets over B-LR, the best classifier. It means that the gain proves insignificant on most datasets.

We also report the mean with the Friedman Rank and variance of AUC scores in Fig. \ref{fig:rank-test-auc-var}. The Friedman Rank and averaged AUC scores shown in Fig. \ref{fig:rank-test-auc-var} (a) and (b) present that superiority of AUC-opt on optimizing AUC score. Variances of AUC shown in (c) and (d) indicate that 0-1 objective optimization (AUC-opt) is more robust, while SVMs are not. Yet, the performance of SVM-Perf is better than SVM but not B-SVM.

\begin{wraptable}{l}{4.3cm}
\small
\vspace{-5mm}
\begin{tabular}{P{0.1\textwidth}|P{0.13\textwidth}}\toprule
Method & Run Time\\\hline 
SVM & 0.1918  $\displaystyle \pm $ 0.4667\\
B-SVM & 0.2862  $\displaystyle \pm $0.5205 \\ 
LR & 0.0785 $\displaystyle \pm $ 0.0489 \\
B-LR & 0.0781 $\displaystyle \pm $0.0491 \\
SPAM & 3.3675 $\displaystyle \pm $ 6.6170 \\ 
SPAUC & 3.3826 $\displaystyle \pm $ 6.7024 \\ 
SVM-Perf & 8.9185 $\displaystyle \pm $ 14.691 \\ 
AUC-opt & 3.7073 $\displaystyle \pm $ 7.4934 \\\bottomrule
\end{tabular}
\vspace{-5mm}
\label{tab:run-time}
\end{wraptable}

\noindent\textbf{Run time.} The left table shows the comparison of the averaged run time (in seconds) over all datasets. The run time of AUC optimizers is slower than standard classifiers. This is because the objective of AUC optimizers is the sum of the loss of training example pairs, which makes them $\mathcal{O}(n^2)$ per batch while AUC-opt uses less on average compared with SVM-Perf.

\noindent\textbf{Comparison of other metrics.} As presented in our supplementary, an interesting observation is that AUC optimizers obtain better performance on balanced accuracy and F1-score than the more traditional classifiers. Results are consistent with the findings shown in \citet{narasimhan2013relationship} where an AUC optimizer is a good way to construct a good balanced classifier. The performance of LR is good on AUC, but not on balanced accuracy and F1 score.

\subsection{Results on real-world datasets}

We test both the approximate AUC optimizers and the standard classification methods on 50 real-world datasets. The AUC scores on testing are reported in Table \ref{tab:auc-scores-real}. First, Balanced Random Forest (B-RF) and Balanced SVM-RBF (B-SVM-RBF) prove the best overall on testing AUC scores. It wins 31.92 and 31.15 datasets on average respectively. This performance is consistent with findings shown in both \citet{fernandez2014we} and \citet{couronne2018random}. Furthermore, a relationship between RF and LR has been studied in \citet{couronne2018random} where it has been shown that RF can obtain much higher AUC scores compared with LR. Boosting-based methods such as AdaBoost\footnote{We treat AdaBoost as the AUC-based method because theoretical finding indicates that AdaBoost is equivalent to RankBoost.} and Gradient Boost (GB) also work well.

RankBoost is inferior to RF and B-RF, winning on only 9 and 8 datasets, respectively. The performance of AdaBoost and RankBoost prove competitive with each other. This has been theoretically justified in \citet{rudin2009margin}. However, interestingly, RankBoost still outperforms AdaBoost over 35 datasets in the testing stage. Gradient Boost (GB) wins more datasets than RankBoost. Generally speaking, two nonlinear AUC optimizers are inferior to these popular nonlinear standard classifiers. This clearly suggests that approximate AUC optimizers may not lead to the best AUC performance, hence having space to improve.  All linear methods lose on average to non-linear ones. 

\section{Conclusion and future work}
\label{section:conclusion}

Our complexity results show linear AUC optimization is NP-complete via reduction to the open hemisphere problem. It remains interesting to prove the hardness results for other hypothesis classes mentioned in \citet{ben2003difficulty}. We then present an optimal method AUC-opt that is both time and space-efficient for optimizing AUC in $\mathbb{R}^2$. We demonstrate that it can be naturally extended to $\mathbb{R}^d$ with a total cost $\mathbb{O}((n_+n_-)^{d-1}\log(n_+ n_-))$. Our empirical results suggest that to justify the objective to optimize AUC, more effort may be needed to improve the optimization quality of AUC optimizers. 

AUC-opt is impractical for real-world datasets since the time complexity is exponentially getting worse 
with the dimension $d$. However, it remains interesting to see whether more efficient algorithms exist for higher but still moderate dimensionality. One potential direction is to use branch and bound as explored in \citet{nguyen2013algorithms}. It is also interesting to compare our method with \citet{rudin2018direct}, a recently proposed method that directly optimizes a rerank statistic. 

\clearpage

\section{Acknowledgments}
The authors would like to thank anonymous reviewers for their valuable comments. The work of Baojian Zhou is partially supported by startup funding from Fudan University. Steven Skiena was partially supported by NSF grants IIS-1926781, IIS-1927227, IIS-1546113, OAC-191952, and a New York State Empire Innovation grant.
\bibliography{aaai23}

\begin{thebibliography}{61}
\providecommand{\natexlab}[1]{#1}

\bibitem[{Ben-David, Eiron, and Long(2003)}]{ben2003difficulty}
Ben-David, S.; Eiron, N.; and Long, P.~M. 2003.
\newblock On the difficulty of approximately maximizing agreements.
\newblock \emph{Journal of Computer and System Sciences}, 66(3): 496--514.

\bibitem[{Ben-Or(1983)}]{ben1983lower}
Ben-Or, M. 1983.
\newblock Lower bounds for algebraic computation trees.
\newblock In \emph{Proceedings of the 15th Annual ACM Symposium on Theory of
  Computing}, 80--86.

\bibitem[{Bradley(1997)}]{bradley1997use}
Bradley, A.~P. 1997.
\newblock The use of the area under the {ROC} curve in the evaluation of
  machine learning algorithms.
\newblock \emph{Pattern Recognition}, 30(7): 1145--1159.

\bibitem[{Brefeld and Scheffer(2005)}]{brefeld2005auc}
Brefeld, U.; and Scheffer, T. 2005.
\newblock {AUC} maximizing support vector learning.
\newblock In \emph{Proceedings of the ICML 2005 workshop on {ROC} Analysis in
  Machine Learning}.

\bibitem[{Breiman(2001)}]{breiman2001random}
Breiman, L. 2001.
\newblock Random forests.
\newblock \emph{Machine learning}, 45(1): 5--32.

\bibitem[{Calders and Jaroszewicz(2007)}]{calders2007efficient}
Calders, T.; and Jaroszewicz, S. 2007.
\newblock Efficient {AUC} optimization for classification.
\newblock In \emph{European Conference on Principles of Data Mining and
  Knowledge Discovery}, 42--53. Springer.

\bibitem[{Chang and Lin(2011)}]{chang2011libsvm}
Chang, C.-C.; and Lin, C.-J. 2011.
\newblock {LIBSVM}: {A} library for support vector machines.
\newblock \emph{ACM Transactions on Intelligent Systems and Technology (TIST)},
  2(3): 1--27.

\bibitem[{Cl{\'e}men{\c{c}}on et~al.(2008)Cl{\'e}men{\c{c}}on, Lugosi, Vayatis
  et~al.}]{clemenccon2008ranking}
Cl{\'e}men{\c{c}}on, S.; Lugosi, G.; Vayatis, N.; et~al. 2008.
\newblock Ranking and empirical minimization of {U}-statistics.
\newblock \emph{The Annals of Statistics}, 36(2): 844--874.

\bibitem[{Cohen, Schapire, and Singer(1998)}]{order1999}
Cohen, W.~W.; Schapire, R.~E.; and Singer, Y. 1998.
\newblock Learning to order things.
\newblock In \emph{NIPS}, 451--457.

\bibitem[{Cortes and Mohri(2004)}]{cortes2004auc}
Cortes, C.; and Mohri, M. 2004.
\newblock {AUC} optimization vs. error rate minimization.
\newblock In \emph{NIPS}, 313--320.

\bibitem[{Cortes and Vapnik(1995)}]{cortes1995support}
Cortes, C.; and Vapnik, V. 1995.
\newblock Support-vector networks.
\newblock \emph{Machine Learning}, 20(3): 273--297.

\bibitem[{Couronn{\'e}, Probst, and Boulesteix(2018)}]{couronne2018random}
Couronn{\'e}, R.; Probst, P.; and Boulesteix, A.-L. 2018.
\newblock Random forest versus logistic regression: a large-scale benchmark
  experiment.
\newblock \emph{BMC Bioinformatics}, 19(1): 270.

\bibitem[{Dobkin and Reiss(1980)}]{dobkin1980complexity}
Dobkin, D.~P.; and Reiss, S.~P. 1980.
\newblock The complexity of linear programming.
\newblock \emph{Theoretical Computer Science}, 11(1): 1--18.

\bibitem[{Dua and Graff(2017)}]{Dua:2019}
Dua, D.; and Graff, C. 2017.
\newblock {UCI} Machine Learning Repository.

\bibitem[{Eban et~al.(2017)Eban, Schain, Mackey, Gordon, Rifkin, and
  Elidan}]{eban2017scalable}
Eban, E.; Schain, M.; Mackey, A.; Gordon, A.; Rifkin, R.; and Elidan, G. 2017.
\newblock Scalable Learning of Non-Decomposable Objectives.
\newblock In \emph{Artificial Intelligence and Statistics}, 832--840.

\bibitem[{Edelsbrunner and Guibas(1989)}]{edelsbrunner1989topologically}
Edelsbrunner, H.; and Guibas, L.~J. 1989.
\newblock Topologically sweeping an arrangement.
\newblock \emph{Journal of Computer and System Sciences}, 38(1): 165--194.

\bibitem[{Elizondo(2006)}]{elizondo2006linear}
Elizondo, D. 2006.
\newblock The linear separability problem: Some testing methods.
\newblock \emph{IEEE Transactions on Neural Networks}, 17(2): 330--344.

\bibitem[{Fawcett(2006)}]{fawcett2006introduction}
Fawcett, T. 2006.
\newblock An introduction to {ROC} analysis.
\newblock \emph{Pattern Recognition Letters}, 27(8): 861--874.

\bibitem[{Feldman et~al.(2012)Feldman, Guruswami, Raghavendra, and
  Wu}]{feldman2012agnostic}
Feldman, V.; Guruswami, V.; Raghavendra, P.; and Wu, Y. 2012.
\newblock Agnostic learning of monomials by halfspaces is hard.
\newblock \emph{SIAM Journal on Computing}, 41(6): 1558--1590.

\bibitem[{Fern{\'a}ndez-Delgado et~al.(2014)Fern{\'a}ndez-Delgado, Cernadas,
  Barro, and Amorim}]{fernandez2014we}
Fern{\'a}ndez-Delgado, M.; Cernadas, E.; Barro, S.; and Amorim, D. 2014.
\newblock Do we need hundreds of classifiers to solve real world classification
  problems?
\newblock \emph{JMLR}, 3133--3181.

\bibitem[{Ferri, Flach, and Hern{\'a}ndez-Orallo(2002)}]{ferri2002learning}
Ferri, C.; Flach, P.~A.; and Hern{\'a}ndez-Orallo, J. 2002.
\newblock Learning Decision Trees Using the Area Under the {ROC} Curve.
\newblock In \emph{ICML}, 139--146.

\bibitem[{Freund et~al.(2003)Freund, Iyer, Schapire, and
  Singer}]{freund2003efficient}
Freund, Y.; Iyer, R.; Schapire, R.~E.; and Singer, Y. 2003.
\newblock An efficient boosting algorithm for combining preferences.
\newblock \emph{JMLR}, 933--969.

\bibitem[{Freund and Schapire(1995)}]{freund1995desicion}
Freund, Y.; and Schapire, R.~E. 1995.
\newblock A desicion-theoretic generalization of on-line learning and an
  application to boosting.
\newblock In \emph{European Conference on Computational Learning Theory},
  23--37. Springer.

\bibitem[{Freund and Schapire(1999)}]{freund1999large}
Freund, Y.; and Schapire, R.~E. 1999.
\newblock Large margin classification using the perceptron algorithm.
\newblock \emph{Machine Learning}, 37(3): 277--296.

\bibitem[{Gao et~al.(2013)Gao, Jin, Zhu, and Zhou}]{gao2013one}
Gao, W.; Jin, R.; Zhu, S.; and Zhou, Z.-H. 2013.
\newblock One-pass {AUC} optimization.
\newblock In \emph{ICML}, 906--914.

\bibitem[{Gao and Zhou(2015)}]{gao2015consistency}
Gao, W.; and Zhou, Z.-H. 2015.
\newblock On the consistency of {AUC} pairwise optimization.
\newblock In \emph{Twenty-Fourth International Joint Conference on Artificial
  Intelligence}, 939–945.

\bibitem[{Gultekin et~al.(2020)Gultekin, Saha, Ratnaparkhi, and
  Paisley}]{gultekin2020mba}
Gultekin, S.; Saha, A.; Ratnaparkhi, A.; and Paisley, J. 2020.
\newblock {MBA}: mini-batch {AUC} optimization.
\newblock \emph{IEEE Transactions on Neural Networks and Learning Systems}.

\bibitem[{Hanley and McNeil(1982)}]{hanley1982meaning}
Hanley, J.~A.; and McNeil, B.~J. 1982.
\newblock The meaning and use of the area under a receiver operating
  characteristic ({ROC}) curve.
\newblock \emph{Radiology}, 143(1): 29--36.

\bibitem[{Herschtal and Raskutti(2004)}]{herschtal2004optimising}
Herschtal, A.; and Raskutti, B. 2004.
\newblock Optimising area under the {ROC} curve using gradient descent.
\newblock In \emph{ICML}, 49.

\bibitem[{Ho and Zimmerman(2006)}]{ho2006number}
Ho, C.; and Zimmerman, S. 2006.
\newblock On the number of regions in an m-dimensional space cut by n
  hyperplanes.
\newblock \emph{Australian Mathematical Society Gazette}, 33(4): 260.

\bibitem[{Joachims(2005)}]{joachims2005support}
Joachims, T. 2005.
\newblock A support vector method for multivariate performance measures.
\newblock In \emph{ICML}, 377--384.

\bibitem[{Johnson and Preparata(1978)}]{johnson1978densest}
Johnson, D.~S.; and Preparata, F.~P. 1978.
\newblock The densest hemisphere problem.
\newblock \emph{Theoretical Computer Science}, 6(1): 93--107.

\bibitem[{Kallus and Zhou(2019)}]{kallus2019fairness}
Kallus, N.; and Zhou, A. 2019.
\newblock The fairness of risk scores beyond classification: Bipartite ranking
  and the x{AUC} metric.
\newblock In \emph{NIPS}, 3433--3443.

\bibitem[{Kot{\l}owski, Dembczy{\'n}ski, and
  H{\"u}llermeier(2011)}]{kotlowski2011bipartite}
Kot{\l}owski, W.; Dembczy{\'n}ski, K.; and H{\"u}llermeier, E. 2011.
\newblock Bipartite ranking through minimization of univariate loss.
\newblock In \emph{ICML}, 1113--1120.

\bibitem[{Le et~al.(2010)Le, Smola, Chapelle, and Teo}]{le2010optimization}
Le, Q.~V.; Smola, A.; Chapelle, O.; and Teo, C.~H. 2010.
\newblock Optimization of ranking measures.
\newblock \emph{JMLR}, 1(2999): 1--48.

\bibitem[{Lee and Preparata(1984)}]{lee1984computational}
Lee, D.-T.; and Preparata, F.~P. 1984.
\newblock Computational geometry? a survey.
\newblock \emph{IEEE Transactions on Computers}, 1072--1101.

\bibitem[{Lei and Ying(2019)}]{lei2019stochastic}
Lei, Y.; and Ying, Y. 2019.
\newblock Stochastic Proximal {AUC} Maximization.
\newblock \emph{arXiv preprint arXiv:1906.06053}.

\bibitem[{Lema{{\^i}}tre, Nogueira, and Aridas(2017)}]{imbalanced-learn}
Lema{{\^i}}tre, G.; Nogueira, F.; and Aridas, C.~K. 2017.
\newblock Imbalanced-learn: A {Python} Toolbox to Tackle the Curse of
  Imbalanced Datasets in Machine Learning.
\newblock \emph{JMLR}, 18(17): 1--5.

\bibitem[{Liu et~al.(2020)Liu, Yuan, Ying, and Yang}]{liu2020stochastic}
Liu, M.; Yuan, Z.; Ying, Y.; and Yang, T. 2020.
\newblock Stochastic auc maximization with deep neural networks.
\newblock In \emph{International Conference on Learning Representations
  (ICLR)}.

\bibitem[{Liu et~al.(2018)Liu, Zhang, Chen, Wang, and Yang}]{liu2018fast}
Liu, M.; Zhang, X.; Chen, Z.; Wang, X.; and Yang, T. 2018.
\newblock Fast Stochastic {AUC} Maximization with $O(1/n)$-Convergence Rate.
\newblock In \emph{ICML}, 3189--3197.

\bibitem[{Maaten and Hinton(2008)}]{maaten2008visualizing}
Maaten, L. v.~d.; and Hinton, G. 2008.
\newblock Visualizing data using t-{SNE}.
\newblock \emph{JMLR}, 2579--2605.

\bibitem[{Menon and Williamson(2016)}]{krishna2016aditya}
Menon, A.~K.; and Williamson, R.~C. 2016.
\newblock Bipartite Ranking: a Risk-Theoretic Perspective.
\newblock \emph{JMLR}, 17(195): 1--102.

\bibitem[{Narasimhan and Agarwal(2013)}]{narasimhan2013relationship}
Narasimhan, H.; and Agarwal, S. 2013.
\newblock On the relationship between binary classification, bipartite ranking,
  and binary class probability estimation.
\newblock In \emph{NIPS}, 2913--2921.

\bibitem[{Natole, Ying, and Lyu(2018)}]{natole2018stochastic}
Natole, M.; Ying, Y.; and Lyu, S. 2018.
\newblock Stochastic proximal algorithms for {AUC} maximization.
\newblock In \emph{ICML}, 3710--3719.

\bibitem[{Nelder and Wedderburn(1972)}]{nelder1972generalized}
Nelder, J.~A.; and Wedderburn, R.~W. 1972.
\newblock Generalized linear models.
\newblock \emph{Journal of the Royal Statistical Society: Series A (General)},
  135(3): 370--384.

\bibitem[{Nguyen and Sanner(2013)}]{nguyen2013algorithms}
Nguyen, T.~T.; and Sanner, S. 2013.
\newblock Algorithms for direct 0-1 loss optimization in binary classification.
\newblock In \emph{ICML}, 1085--1093.

\bibitem[{Oshiro, Perez, and Baranauskas(2012)}]{oshiro2012many}
Oshiro, T.~M.; Perez, P.~S.; and Baranauskas, J.~A. 2012.
\newblock How many trees in a random forest?
\newblock In \emph{International Workshop on Machine Learning and Data Mining
  in Pattern Recognition}, 154--168. Springer.

\bibitem[{Pedregosa et~al.(2011)Pedregosa, Varoquaux, Gramfort, Michel,
  Thirion, Grisel, Blondel, Prettenhofer, Weiss, Dubourg, Vanderplas, Passos,
  Cournapeau, Brucher, Perrot, and Duchesnay}]{scikit-learn}
Pedregosa, F.; Varoquaux, G.; Gramfort, A.; Michel, V.; Thirion, B.; Grisel,
  O.; Blondel, M.; Prettenhofer, P.; Weiss, R.; Dubourg, V.; Vanderplas, J.;
  Passos, A.; Cournapeau, D.; Brucher, M.; Perrot, M.; and Duchesnay, E. 2011.
\newblock Scikit-learn: Machine Learning in {P}ython.
\newblock \emph{JMLR}, 12: 2825--2830.

\bibitem[{Probst and Boulesteix(2017)}]{probst2017tune}
Probst, P.; and Boulesteix, A.-L. 2017.
\newblock To tune or not to tune the number of trees in random forest.
\newblock \emph{JMLR}, 18(1): 6673--6690.

\bibitem[{Provost and Fawcett(2001)}]{provost2001robust}
Provost, F.; and Fawcett, T. 2001.
\newblock Robust classification for imprecise environments.
\newblock \emph{Machine Learning}, 42(3): 203--231.

\bibitem[{Qin, Liu, and Li(2010)}]{qin2010general}
Qin, T.; Liu, T.-Y.; and Li, H. 2010.
\newblock A general approximation framework for direct optimization of
  information retrieval measures.
\newblock \emph{Information retrieval}, 13(4): 375--397.

\bibitem[{Rakotomamonjy(2004)}]{rakotomamonjy2004optimizing}
Rakotomamonjy, A. 2004.
\newblock Optimizing {AUC} with support vector machine.
\newblock In \emph{European Conference on Artificial Intelligence Workshop on
  {ROC} Curve and AI}, 469--478.

\bibitem[{Rudin and Schapire(2009)}]{rudin2009margin}
Rudin, C.; and Schapire, R.~E. 2009.
\newblock Margin-based ranking and an equivalence between AdaBoost and
  RankBoost.
\newblock \emph{JMLR}, 2193--2232.

\bibitem[{Rudin and Wang(2018)}]{rudin2018direct}
Rudin, C.; and Wang, Y. 2018.
\newblock Direct Learning to Rank And Rerank.
\newblock In \emph{Proceedings of the Twenty-First International Conference on
  Artificial Intelligence and Statistics}, volume~84, 775--783. PMLR.

\bibitem[{Steck(2007)}]{steck2007hinge}
Steck, H. 2007.
\newblock Hinge rank loss and the area under the {ROC} curve.
\newblock In \emph{European Conference on Machine Learning}, 347--358.
  Springer.

\bibitem[{Vogel et~al.(2021)Vogel, Bellet, Cl{\'e}men
  et~al.}]{vogel2021learning}
Vogel, R.; Bellet, A.; Cl{\'e}men, S.; et~al. 2021.
\newblock Learning Fair Scoring Functions: Bipartite Ranking under ROC-based
  Fairness Constraints.
\newblock In \emph{International Conference on Artificial Intelligence and
  Statistics}, 784--792. PMLR.

\bibitem[{Yan et~al.(2003)Yan, Dodier, Mozer, and
  Wolniewicz}]{yan2003optimizing}
Yan, L.; Dodier, R.~H.; Mozer, M.; and Wolniewicz, R.~H. 2003.
\newblock Optimizing classifier performance via an approximation to the
  Wilcoxon-Mann-Whitney statistic.
\newblock In \emph{ICML}, 848--855.

\bibitem[{Yang and Ying(2022)}]{yang2022auc}
Yang, T.; and Ying, Y. 2022.
\newblock AUC maximization in the era of big data and AI: A survey.
\newblock \emph{ACM Computing Surveys (CSUR)}.

\bibitem[{Ying, Wen, and Lyu(2016)}]{ying2016stochastic}
Ying, Y.; Wen, L.; and Lyu, S. 2016.
\newblock Stochastic online {AUC} maximization.
\newblock In \emph{NIPS}, 451--459.

\bibitem[{Yogananda, Murty, and Gopal(2007)}]{linear2007}
Yogananda, A.; Murty, M.~N.; and Gopal, L. 2007.
\newblock A fast linear separability test by projection of positive points on
  subspaces.
\newblock In \emph{Proceedings of the 24th International Conference on Machine
  Learning (ICML-07)}, 713--720.

\bibitem[{Zhao et~al.(2011)Zhao, Hoi, Jin, and Yang}]{zhao2011online}
Zhao, P.; Hoi, S.~C.; Jin, R.; and Yang, T. 2011.
\newblock Online {AUC} maximization.
\newblock In \emph{ICML}, 233--240.

\end{thebibliography}

\appendix
\onecolumn
\begin{center}
\textbf{Does it pay to optimize AUC? --- Supplementary Materials}    
\end{center}

\paragraph{Supplementary Outline.} \S\ref{section:proofs} presents the detailed proof of Theorem 1. Our generalization algorithm, namely AUC-opt$(\mathcal{D}, d)$ on $\mathbb{R}^d$ and its time complexity analysis are in \S\ref{section:generalization}. Baseline methods including their parameter tunings are presented in \S\ref{section:baselines}. Datasets and more experimental results are given in \S\ref{section:datasets} and \S\ref{section:experiments}, respectively.

\section{The proof of Theorem 1}
\label{section:proofs}

Recall that we have the following LAO problem and the theorem of its NP-completeness.

\begin{problem}[Linear AUC Optimization (LAO)]
Given the dataset $\mathcal{D}$ and the hypothesis class $\mathcal{H}:=\{\bm w: \bm w \in \mathbb{R}^d\}$, the LAO problem is to find a $\bm w\in \mathcal{H}$ such that the empirical AUC score is maximized, that is
\begin{small}
\begin{equation}
(\textit{LAO}) \quad \bm w^\star \in \argmax_{\bm w \in \mathcal{H}} \sum_{i=1}^{n_+} \sum_{j=1}^{n_-} \frac{{\bm 1} \left[\bm w^\top \bm x_i^{+}> \bm w^\top \bm x_j^{-} \right]}{n_+ n_-}, 
\end{equation}
\end{small}
where the indicator ${\bm 1}\left[A\right]=1$ if $A$ is true 0 otherwise.
\end{problem}

\begin{theorem}[NP-complete of LAO]
Consider the linear AUC optimization problem defined in Problem \ref{problem:lao-problem}, if $\mathcal{D}$ is not linearly separable, LAO is NP-complete when both $n$ and $d$ are arbitrary.
\end{theorem}

To prove the NP-completeness, we reformulate the LAO problem to its feasibility problem and then prove its NP-completeness. Before we prove Theorem \ref{thm:np-complete-lao}, let us introduce the spherical separation problem and a corresponding lemma, which shows that the spherical separation problem is in NP.

\begin{problem}[Spherical Separation (SS)]
Given two sets of points $\mathcal{U} :=\{\bm u_1,\bm u_2,\ldots, \bm u_n\}\subseteq \mathcal{S}^{d-1} \cap \mathbb{R}^d$ and $\mathcal{V} :=\{ \bm v_1, \bm v_2,\ldots, \bm v_m\}\subseteq \mathcal{S}^{d-1} \cap \mathbb{R}^d$. Does there exist a hyperplane $\bm w^\prime$ such that $\mathcal{U}$ is separable from $\mathcal{V}$?  \label{prob:ss}
\end{problem}

\begin{lemma}[\cite{dobkin1980complexity}] SS problem defined above is in NP, i.e., SS $\in $ NP.
\label{lemma:ss-np}
\end{lemma}

\begin{proof}
Without loss of the generality, let us consider the LAO problem in $\mathcal{H} := \mathbb{Q}^d$. Our first step is to reformulate the LAO problem as a feasibility problem (F-LAO) shown in Equation~(\ref{inequ:feasbility-lao1}). The goal is to prove this feasibility problem is NP-complete. It has two steps: 1) to show F-LAO is in NP (i.e., nondeterministic polynomial-time); and 2) to prove F-LAO is NP-hard.
\begin{center}
 \begin{minipage}{.47\textwidth}
\begin{tcolorbox}[colback=white,left=0pt,right=0pt,top=0pt,bottom=0pt,outer arc=0pt,arc=0pt]
\textbf{Feasibility problem of LAO (F-LAO)\quad} 

\textbf{Input:} Given a finite dataset $\mathcal{D}$ and a set of linear functions $\mathcal{H}$ and positive integers $d, t$.

\textbf{Question:} Does there exist $ \bm w \in \mathcal{H}$ such that 
\begin{equation}
\sum_{i=1}^{n_+} \sum_{j=1}^{n_-} {\bm 1} \left[\bm w^\top \bm x_{i}  > \bm w^\top \bm x_{j} \right] \geq t \ ? \label{inequ:feasbility-lao1}
\end{equation}
\end{tcolorbox}
\end{minipage}\quad\quad\qquad
\begin{minipage}{.42\textwidth}\centering

\includegraphics{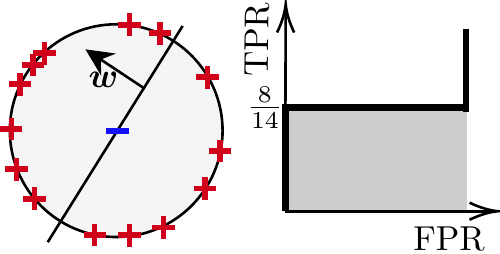}
\captionof{figure}{Left: An instance of OH. Right: the AUC curve of an instance of LAO.}
\label{fig:illustration-lao-oh1}
\end{minipage}  
\end{center}

\paragraph{Claim 1: F-LAO is in NP.} To prove it is in NP, we need to show that 1) there exists an algorithm (verifier) verifying certificates in polynomial time; and 2) the certificate length is polynomial in the input size. For the first step, given any $\bm w \in \mathcal{H}$, one can find a polynomial-time verifier such that it finishes in $\mathcal{O}(n_p n_q d)$ time. To do so, we propose a polynomial-time verifier presented as in Algo. \ref{alg:verifier}, which finishes in polynomial-time, i.e. $\mathcal{O}(n_p n_q d)$. The rest is to show the certificate $\bm w$ could be encoded in polynomial with respect to the input size. All we need to do is to show that there exists another solution $\bm w^\prime$ such that $\bm w^\prime$ could be encoded in polynomial with respect to the input size. Let $\mathcal{U} := \{\bm x_{i}^+: {\bm w^\prime}^\top \bm x_{i}^+ > {\bm w^\prime}^\top \bm x_{j}^- \}$ and $\mathcal{V} := \{\bm x_j^-: {\bm w^\prime}^\top \bm x_{i}^+ <= {\bm w^\prime}^\top \bm x_{j}^-\}$. To find such $\bm w^\prime$ is equivalent to solve SS shown in Problem \ref{prob:ss}. By Lemma \ref{lemma:ss-np} and the definition of $\mathcal{U}$, $\mathcal{V}$, we see that such $\bm w^\prime$ can be encoded in polynomial length with respect to the input since SS problem is in NP. Therefore, $\bm w = \bm w^\prime$ can also be encoded in polynomial.

\begin{algorithm}[t]
\textbf{Input: } $\mathcal{D}, d, t$ and the certificate $\bm w$
\caption{Verifier$(\left\langle\mathcal{D}, d, t \right\rangle, \bm w )$}
\begin{algorithmic}[1]
\STATE $c \leftarrow 0$ \hfill $\triangleright$ The counter of the number of inequalities
\FOR{$(i, j) \in \{1,\ldots, n_+\} \times \{1,\ldots, n_-\}$} 
\IF{$ \bm w^\top \bm x_{i}^+  >  \bm w^\top \bm x_{j}^- $}
\STATE $ c \leftarrow c + 1$
    \IF{$c \geq t$}
    \STATE \textbf{Return } True
    \ENDIF
\ENDIF
\ENDFOR
\STATE \textbf{Return } False
\end{algorithmic}
\label{alg:verifier}
\end{algorithm}

\paragraph{Claim 2: F-LAO is NP-hard.} We need to show that given any instance of an OH problem, it can be solved by an instance of LAO problem. To do this, we show that any instance of the OH problem can be transformed into an instance of the feasibility problem defined in (\ref{inequ:feasbility-lao1}). First of all, we need to construct a training dataset $\mathcal{D}$ so that an instance of Problem (\ref{equ:linear-auc-posi}) can be defined. Notice that we can rewrite vectors in $\mathcal{D}$ and construct new vectors $\bm x_{i}^+$ and $\bm x_{1}^-$ as the following
\begin{align*}
\bm s_1 = \underbrace{(\bm s_1 + \bm x_{1}^-)}_{\bm x_{1}^+} - \bm x_1^-, \bm s_2 = \underbrace{(\bm s_2 + \bm x_1^-)}_{\bm x_{2}^+} - \bm x_1^-, \ldots, \bm s_t = \underbrace{(\bm s_t + \bm x_1^-)}_{\bm x_{t}^+} - \bm x_1^-.
\end{align*}
We construct the training labels such that $y_{1}^+, y_{2}^+, \ldots, y_{t}^+$ are 1, and $y_{1}^-=-1$, then the constructed $\mathcal{D}:=\{(\bm x_{1}^+,y_{1}^+),\ldots, (\bm x_{t}^+,y_{t}^+), (\bm x_1^-,y_{1}^-)\}$. Fig. \ref{fig:illustration-lao-oh1} illustrates this correspondence where each positive sign is $\bm s_i$, and the negative sign is $\bm x_1^- = \bm 0$. The OH problem has 14 inequalities shown on the left while the normal vector $\bm w$ defines a hyperplane which corresponds to $\sum_{i=1}^{n_+} \sum_{j=1}^{n_-} {\bm 1}[\bm w^\top \bm x_{i}^+ > \bm w^\top \bm x_{j}^-] = 8$ as shown on the right of Fig. \ref{fig:illustration-lao-oh1}. To check whether there exists $\bm w$ so that the number of inequalities can be satisfied is at least $m$, one only needs to check whether there exists a $\bm w$ and $t = m$ such that~(\ref{inequ:feasbility-lao1}) is satisfied. This transformation can be done in polynomial time. Thus any solution of this instance of the LAO problem defined by $\mathcal{D}$ is a solution of OH. Therefore, the feasibility problem is NP-hard.
\end{proof} 

\paragraph{The implementation of AUC-opt on $\mathbb{R}^2$.} We implemented AUC-opt by using C11 with a Python-3.7 wrapper. As described in Algorithm 1, AUC-opt has a parameter $\epsilon$. This parameter is a small enough constant so that it is smaller than the minimal difference between any two consecutive slopes; that is, $\epsilon$ satisfies
\begin{equation}
\epsilon < \min_{(i,j) \ne (i^\prime,j^\prime)} \left| \frac{ x_{i2} -  x_{j2}}{ x_{i1} - x_{j1}} -\frac{ x_{i^\prime2} - x_{j^\prime2}}{ x_{i^\prime1} - x_{j^\prime1}} \right|. \nonumber
\end{equation}
In our implementation, we dynamically set it as the half of the difference between any two consecutive slopes. In other words, we sort the slope by ascending order, $\ldots,s_i,s_{i+1},\ldots$ where each $\epsilon_i = (s_{i+1} - s_i) / 2$.

\section{Generalization to $\mathbb{R}^d$: AUC-opt$(\mathcal{D}, d)$}
\label{section:generalization}

The idea of projecting the original problem into lower-dimensional subproblems is mainly inspired by \citet{johnson1978densest}. To simplify our problem, we can map pair of points $(\bm x_i^+, \bm x_j^-)$ onto sphere $\mathcal{S}^{d-1}$.  Let us consider the function $f_{i,j}(\bm w) := \bm w^\top (\bm x_i^+ - \bm x_j^-) / \|\bm x_i^+ - \bm x_j^-\|_2$.\footnote{Notice that any pair of points such that $\bm x_i^+ = \bm x_j^-$ is not interesting since it never contributes to the AUC score. Here we assume $\bm x_i^+ \ne \bm x_j^-$ for $i = 1, 2, \ldots, n_+$ and $j = 1, 2, \ldots, n_-$.} Our goal is essentially to find such $\bm w$ so that as many $f_{i,j} > 0$ as possible.  Define such nonzero points $\bm p_t := (\bm x_i^+ - \bm x_j^-) / \|\bm x_i^+ - \bm x_j^-\|_2$ where $t = 1, 2, \ldots, n_+ n_-$. Our LAO problem can be reformulated as the following problem: 

Given a set of nonzero points $\mathcal{K} := \{\bm p_t: \bm p_t, t = 1, 2, \ldots, n_+ n_-\}$ on $\mathcal{S}^{d-1}$, find a hyperplane $\bm w$ crossing the origin so that $\sum_{t=1}^{n_+ n_-} {\bm 1}\left[\bm w^\top \cdot \bm p_t > 0\right]$ is maximized.

Follow the idea of finding densest hemisphere from \cite{johnson1978densest}, let the hyperplane crossing the origin defined by $\bm p_t$ as $H_t := \{\bm x: \bm x^\top\cdot\bm p_t = 0, \bm x \in \mathbb{R}^d\}$. It is enough to consider $\bm w \in \cup_{t=1,\ldots,n_+n_-} H_t$ since $H_t$ divides $\mathbb{R}^d$ into at most $2 \left[ {n-1 \choose 0}, {n-1 \choose 1},\ldots, {n-1 \choose d-1} \right]$ cones \citep{ho2006number} and each cone defines an equivilent class.\footnote{For any $\bm x$ and $\bm y$ in the same cone, we have $\sum_{t=1}^{n_+ n_-} {\bm 1}\left[\bm x^\top \cdot \bm p_t > 0\right] = \sum_{t=1}^{n_+ n_-} {\bm 1}\left[\bm y^\top \cdot \bm p_t > 0\right]$.} Let $\bm w^\prime$ be a point on the face of a specific cone such that the AUC score is maximized. Since $\bm w^\prime$ is on the face, $\bm w^\prime$ must be on a specific hyperplane $H_{t^\prime}$. Therefore, it is enough to iterate all such hyperplane and solving the problem recursively on these hyperplanes. Specifically, we project points onto $H_{t^\prime}$ and then solve problem in $d-1$ dimensional subspace (changing the number of coordinates from $d$ to $d-1$). Let the projection be defined as $\bm P(\bm x) := \bm x - (\bm x^\top \cdot \bm u / \| \bm u\|^2) \cdot \bm u$ where $\bm u$ is the normal vector. The critical property of $\bm P$ is that for any $\bm x \in H(\bm u)$, $\bm x^\top \bm P(\bm p_t) = \bm x^\top \cdot \bm p_t$.  Therefore, the inner products of $\bm p_t$ with $\bm x$ are the same as those of projected ones. We summarize the proposed method as in Algo. \ref{algo:auc-opt-3d}. This recursive procedure projects points onto subspaces until the dimension reduces to 2.

\begin{algorithm}[H]
\small
\caption{$[\operatorname{AUC}_{\operatorname{opt}},\bm w^\star]=$AUC-opt($\mathcal{D},d$)}
\begin{algorithmic}[1]
\STATE $\mathcal{K} = \{(\bm x_i^+ - \bm x_j^-): i \in \{1,\ldots,n_+\},j \in \{1,\ldots,n_-\}\}$
\IF {d=2}
\STATE  \textbf{return} $\operatorname{AUC}_{cur}, \bm w^\prime = \operatorname{AUC-opt}(\mathcal{D})$ \hfill $\triangleright$ call AUC-opt
\ENDIF
\FOR{$\bm u \in \mathcal{K}$}
\STATE $\mathcal{P} = \left\{\left(\bm x - \frac{\bm x^\top \cdot \bm u}{\| \bm u\|^2} \cdot \bm u, y\right): (\bm x, y) \in \mathcal{D}\right\}$ \hfill $\triangleright$ project points of $\mathcal{D}$ onto the hyperplane defined by $\bm u$.
\STATE $\mathcal{P}^\prime = \operatorname{change\_coordinates}(\bm u, \mathcal{P})$ \hfill $\triangleright$ change coordinates so that points are presented in $d-1$ coordinates.
\STATE $\operatorname{AUC}_{cur}, \bm w^\prime  =\operatorname{AUC-opt}(\mathcal{P}^\prime, d-1)$
\STATE $ \bm w = \operatorname{change\_coordinates\_back}(\bm u, \bm w^\prime )$ \hfill $\triangleright$ change $d-1$ coordinates of $\bm w^\prime$ back to $d$ coordinates.
\IF {$\operatorname{AUC}_{\operatorname{opt}} < \operatorname{AUC}_{\operatorname{cur}}$}
\STATE $\bm w^\star = \bm w$, $\operatorname{AUC}_{\operatorname{opt}} = \operatorname{AUC}_{\operatorname{cur}}$
\ENDIF
\ENDFOR
\\\hrulefill
\STATE \textbf{procedure}\quad $\mathcal{P}^\prime =\text{change\_coordinates}(\bm u, {\mathcal{P}})$
\STATE $\mathcal{P}^\prime = \{\}$
\FOR{$\bm p \in \mathcal{P}$}
\STATE find unit vectors $\bm u_1, \bm u_2, \ldots, \bm u_{d-1}$ such that $\bm u_i ^\top \bm u = 0$.
\STATE $\mathcal{P}^\prime = \mathcal{P}^\prime \cup [\bm p^\top \bm u_1,\bm p^\top \bm u_2, \ldots, \bm p^\top \bm u_{d-1}]^\top$
\ENDFOR
\STATE \textbf{Return} $\mathcal{P}^\prime$
\\\hrulefill
\STATE \textbf{procedure}\quad $\bm w =\text{change\_coordinates\_back}(\bm u, \bm w^\prime)$
\STATE let $\bm w^\prime = [w_1^\prime, w_2^\prime, \ldots, w_{d-1}^\prime]^\top$
\STATE find unit vectors $\bm u_1, \bm u_2, \ldots, \bm u_{d-1}$ such that $\bm u_i ^\top \bm u = 0$.
\STATE $\bm w = w_1^\prime \bm u_1 + w_2^\prime \bm u_2 + \ldots w_{d-1}^\prime \bm u_{d-1}$
\STATE \textbf{Return} $\bm w$
\end{algorithmic}
\end{algorithm}

\begin{theorem}[Time complexity of AUC-opt($\mathcal{D},d$)]
Given the dataset $\mathcal{D} := \{(\bm x_i, y_i): i \in \left\{1,2,\ldots,n\right\}\}$ where $\bm x_i \in \mathbb{R}^d$ and $y_i \in \{\pm 1\}$ and a collection of linear separators $\mathcal{H}:=\{\bm w: \bm w\in\mathbb{R}^d\}$. There exists an algorithm solves the LAO problem~(\ref{equ:linear-auc-posi}) in $\mathcal{O}\left((n_+ n_-)^{d-1}\log (n_+ n_-)\right)$.
\end{theorem}
\begin{proof}
The correctness of the proposed algorithm has been described above. The main loop of Algo. \ref{algo:auc-opt-3d} has $|\mathcal{K}| = n_+ n_-$ iterations in total where each iteration complexity is the total complexity of AUC-opt($\mathcal{P}^\prime, d - 1$). Clearly, the total run time of AUC-opt($\mathcal{P}^\prime, d$) is $\mathcal{O}((n_+ n_-)\cdot (n_+ n_-) \cdots (n_+ n_-) \log (n_+ n_-)) = \mathcal{O}\left((n_+ n_-)^{d-1}\log (n_+ n_-)\right)$.
\end{proof}
\begin{remark}
Our method essentially follows the idea of \cite{johnson1978densest} where we map interesting points $(\bm x_i^+ - \bm x_j^-)$ onto the sphere $\mathcal{S}^{d-1}$ and then solve the problem recursively. However, at the lowest level $d=2$, we call AUC-opt to obtain the best classifier for getting the best AUC score.
\end{remark}

\section{Descriptions of baseline methods}
\label{section:baselines}

\paragraph{Random Forest (RF).} For each specific Random Forest (RF) classifier \citep{breiman2001random}, there are two important parameters: 1) the number of trees; and 2) the number of features. Two common strategies are used to choose the number of features: 1) the number of trees is set to $\log(d)+1$ and 2) the number of features is set to $\sqrt{d}$ where we choose $\sqrt{d}$. The number of tree chosen in our experiments are from $\{10, 50, 100, 150, 200, 300, 400, 500\}$. Notice that the number of trees does not necessarily increase the AUC score as reported in \cite{oshiro2012many}. Both \cite{oshiro2012many} and \cite{probst2017tune} show that the performance of AUC score has no big gain after 128 trees. We use the implementation of RF from the sklearn package \citep{scikit-learn}.

\paragraph{Balanced RF (B-RF).} B-RF is an RF classifier where the class weight is set to be ``balanced'' while the other parameter setting is the same as RF's.

\paragraph{RankBoost.} The RankBoost algorithm is proposed in \citet{freund2003efficient}. There are different strategies for selecting weak learners. As suggested in \citet{cortes2004auc}, we choose the third method for selecting the weak learners. Because there are no missing values in our datasets, the parameter $q_{\text{def}}$  of weak learners defined is set to $0$. The only parameter for RankBoost is the number of iterations which is tuned from $\{10, 50, 100, 200, 300, 500, 800, 1000\}$. We implement RankBoost in C11 language with a Python-3.7 wrapper. 

\paragraph{AdaBoost.} The AdaBoost method \citep{freund1995desicion} has one important parameter which is the maximum number of estimators at which boosting is terminated. Specifically, this parameter is chosen from $\{2, 3, 4, 5, 6, 10, 50, 100, 150, 200\}$ while the default parameter is 50. 

\paragraph{Logistic Regression (LR).} For the logistic regression classifier \citep{nelder1972generalized}, there is an $\ell_2$ regularization parameter. We did not select $\ell_1$ regularization is because, in almost all of the datasets, the number of training samples is much larger than the data dimension. Hence, the models are assumed to be dense. The parameter space we tried is from $c\in \left[10^{-6},10^6\right]$ by using 5-fold cross-validation.

\paragraph{Balanced LR (B-LR).} B-LR is the LR classifier when the class weight is set to be ``balanced''.

\paragraph{Support Vector Machine (SVM).} For the Support Vector Machine (SVM) classifier \citep{cortes1995support}, we use the standard setting of linear SVM where the parameter of $\ell_2$-regularization has been tuned from $c \in [10^{-6},10^6]$. 

\paragraph{Balanced SVM (B-SVM).} B-SVM is the SVM classifier when the class weight is set to be ``balanced''.

\paragraph{Gradient Boosting (GB).} The GB algorithm is an additive model where a single regression tree is induced. Similar to AdaBoost, we tune the number of estimators from $\{10, 50, 100, 200, 300, 500, 800, 1000\}$. 

\paragraph{SVM with Radial Basis Function kernel (SVM-RBF).} The SVM-RBF classifier is the support vector classification method using the nonlinear kernel, namely the Radial Basis Function (RBF) kernel. The parameter $c$ is tuned from $[10^{-6},10^6]$. The degree is set to a default value of 3. 

\paragraph{Balanced SVM-RBF.} B-SVM-RBF is the SVM-RBF classifier when the class weight is set to be ``balanced''.

\paragraph{SVM-Perf.} The SVM-Perf method is proposed in \citet{joachims2005support} where the goal is to minimize the AUC loss by using the SVM-based method. We downloaded the code from \url{http://www.cs.cornell.edu/people/tj/svm_light/svm_perf.html}. The parameter $c$ is tuned from $[2^{-20},2^{11}]$. Other parameters are set to default values. 

\paragraph{Stochastic Proximal AUC maximization (SPAUC).} The SPAUC maximization algorithm is proposed in \citet{lei2019stochastic}. SPAUC has two parameters: 1) The learning rate $\eta_t := 2/(\mu t + 1)$, which involves a parameter $\mu$ where it is tuned from $\mu \in [10^{-7}, 10^{-2.5}]$ as suggested in the paper; and 2) We use the $\ell^2$-norm as the regularization where the regularization parameter is tuned from $[10^{-5}, 10^5]$. The maximal number of epochs is fixed to 100. The Python interface of the SPAUC algorithm is as the following. We implemented SPAUC by using C11 with a Python-3.7 wrapper.

\paragraph{SPAM.} The Stochastic Proximal AUC Maximization (SPAM) algorithm is proposed in \citet{natole2018stochastic}. We use $\ell_2$ regularization where the regularization parameter is tuned from $[10^{-5}, 10^5]$ as suggested in the paper. Since there exists a constant coefficient $c$ over the learning rate, i.e. $c/\sqrt{t}$, we tune it from $[1,100]$. The maximal number of epochs is fixed to 100. The interface of SPAM is as the following. We implemented SPAM by using C11 with a Python-3.7 wrapper. 

\clearpage
\section{Datasets}
\label{section:datasets}
\def\arraystretch{1.5}%
\begin{center}
\begin{table}[t]
\caption{Binary classification datasets. Datasets are collected from \citet{imbalanced-learn,chang2011libsvm,Dua:2019}. For dataset x[y] in each dataset name column, x is a multi-classification dataset.  For such dataset, x[y] means samples of label y will be treated as positive samples otherwise negative samples.}
\centering
\begin{tabular}{P{0.15\textwidth}|P{0.1\textwidth}|P{0.1\textwidth}|P{0.1\textwidth}|P{0.15\textwidth}|P{0.1\textwidth}|P{0.1\textwidth}|P{0.1\textwidth}}
\hline\hline
dataset name & $n$ & $d$ & $\displaystyle n_p /n$ & dataset name & $n$ & $d$ & $\displaystyle n_p /n$ \\\hline 
splice & 3175 & 60 & 0.5191 & vowel[hid] & 990 & 10 & 0.0909 \\\hline 
mushrooms & 8124 & 112 & 0.4820 & spectrometer & 531 & 93 & 0.0847 \\\hline 
banana & 5300 & 2 & 0.4483 & cardio[3] & 2126 & 21 & 0.0828 \\\hline 
australian & 690 & 14 & 0.4449 & car-eval[34] & 1728 & 21 & 0.0775 \\\hline 
spambase & 4601 & 57 & 0.3940 & isolet & 7797 & 617 & 0.0770 \\\hline 
ionosphere & 351 & 34 & 0.3590 & us-crime & 1994 & 100 & 0.0752 \\\hline 
fourclass & 862 & 24 & 0.3561 & yeast[ml8] & 2417 & 103 & 0.0736 \\\hline 
breast-cancer & 683 & 10 & 0.3499 & scene & 2407 & 294 & 0.0735 \\\hline 
pima & 768 & 8 & 0.3490 & libras-move & 360 & 90 & 0.0667 \\\hline 
yeast[cyt] & 1484 & 8 & 0.3120 & seismic & 2584 & 24 & 0.0658 \\\hline 
german & 1000 & 24 & 0.3000 & thyroid-sick & 3772 & 52 & 0.0612 \\\hline 
svmguide3 & 1284 & 22 & 0.2625 & coil-2000 & 9822 & 85 & 0.0597 \\\hline 
vehicle[bus] & 846 & 18 & 0.2577 & arrhythmia[06] & 452 & 278 & 0.0553 \\\hline 
vehicle[saab] & 846 & 18 & 0.2565 & solar-flare[m0] & 1389 & 32 & 0.0490 \\\hline 
vehicle[van] & 846 & 18 & 0.2352 & oil & 937 & 49 & 0.0438 \\\hline 
spectf & 267 & 44 & 0.2060 & letter[a] & 20000 & 16 & 0.0394 \\\hline 
ecoli[imu] & 336 & 7 & 0.1042 & car-eval[4] & 1728 & 21 & 0.0376 \\\hline 
pen-digits[0] & 10992 & 16 & 0.1040 & wine-quality & 4898 & 11 & 0.0374 \\\hline 
page-blocks & 5473 & 10 & 0.1023 & letter[z] & 20000 & 16 & 0.0367 \\\hline 
opt-digits[0] & 5620 & 64 & 0.0986 & yeast[me2] & 1484 & 8 & 0.0344 \\\hline 
opt-digits[8] & 5620 & 64 & 0.0986 & yeast[me1] & 1484 & 8 & 0.0296 \\\hline 
satimage[4] & 6435 & 36 & 0.0973 & ozone-level & 2536 & 72 & 0.0288 \\\hline 
pen-digits[5] & 10992 & 16 & 0.0960 & w7a & 34780 & 300 & 0.0282 \\\hline 
abalone[7] & 4177 & 10 & 0.0936 & mammography & 11183 & 6 & 0.0232 \\\hline 
sick-euthyroid & 3163 & 42 & 0.0926 & abalone[19] & 4177 & 10 & 0.0077\\\hline
\end{tabular}
\label{tab:datasets}
\end{table}	
\end{center}

\begin{figure}
\centering
\begin{minipage}{0.45\textwidth}
\includegraphics[scale=.8]{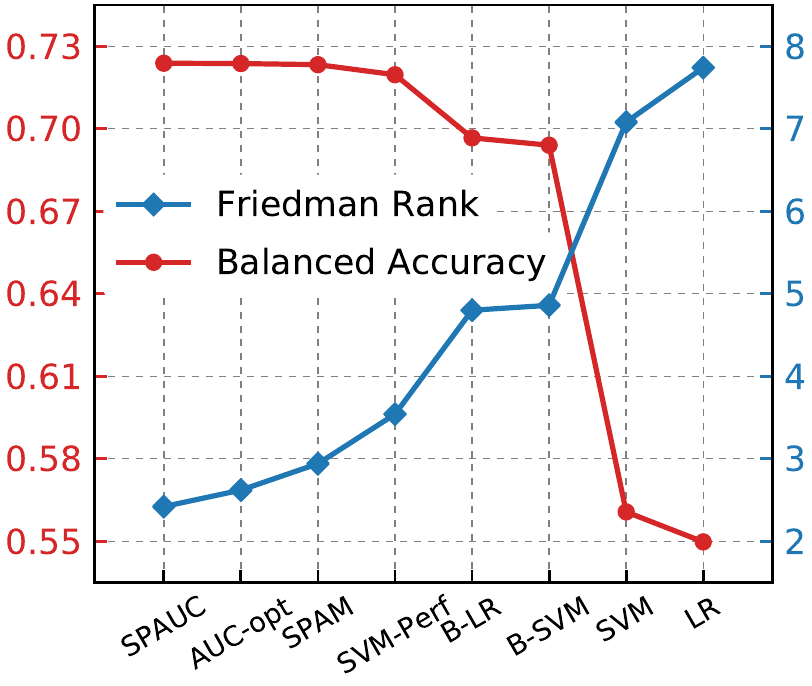}
\end{minipage}\qquad\qquad
\begin{minipage}{0.45\textwidth}
\includegraphics[scale=.8]{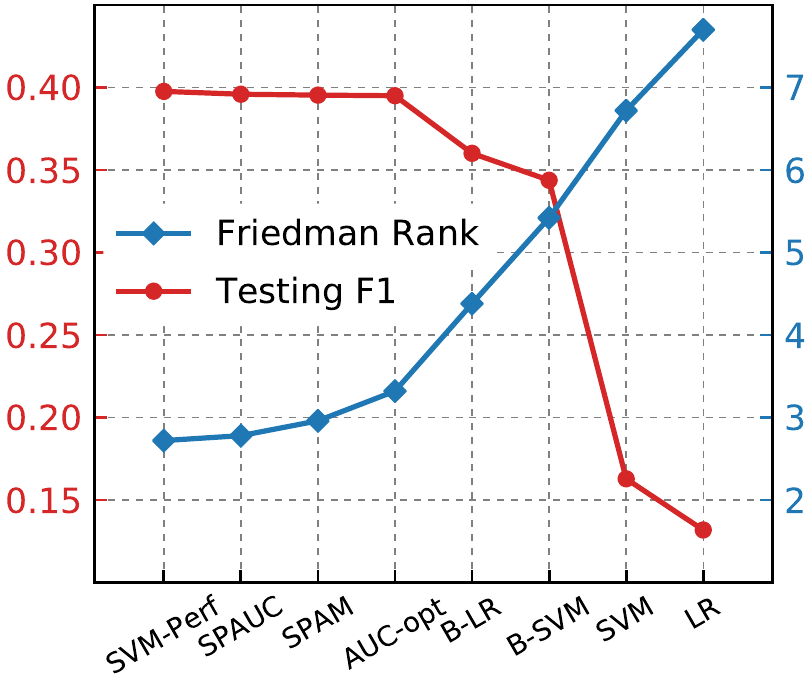}
\end{minipage}
\caption{Balanced accuracies (Left) and F1 scores (Right) as a function of classifiers on testing datasets.}
\label{fig:test-b-acc-score}
\end{figure}

\paragraph{Real-world datasets.} Table \ref{tab:datasets} lists 50 datasets where  they are ordered by imbalanced-ratio $\gamma:=n_p / n$.  All datasets are shuffling 200 times. For each time of the shuffling, 50\% of samples are used for training and the rest are used for testing. All dataset samples have been standardized. Specifically, for each dataset, training and testing samples are standardized by the mean and standard deviation of training samples.

\paragraph{Synthetic datasets using t-SNE.} We project 50 datasets onto both $\mathbb{R}^2$ and $\mathbb{R}^3$ using t-SNE \citep{maaten2008visualizing} so that class patterns are conserved. Projected points will be used as training and testing samples while keep labels unchanged. We use a sklearn implementation of t-SNE. The parameter of perplexity is set to 50 as suggested while all other parameters are set to default values. For the simulation datasets in $\mathbb{R}^3$, we restricted the number of training samples within $(n_+ n_- = 2000)$ to make our experiments fast.

\begin{figure}
\centering
\begin{minipage}{0.45\textwidth}
\includegraphics[scale=.8]{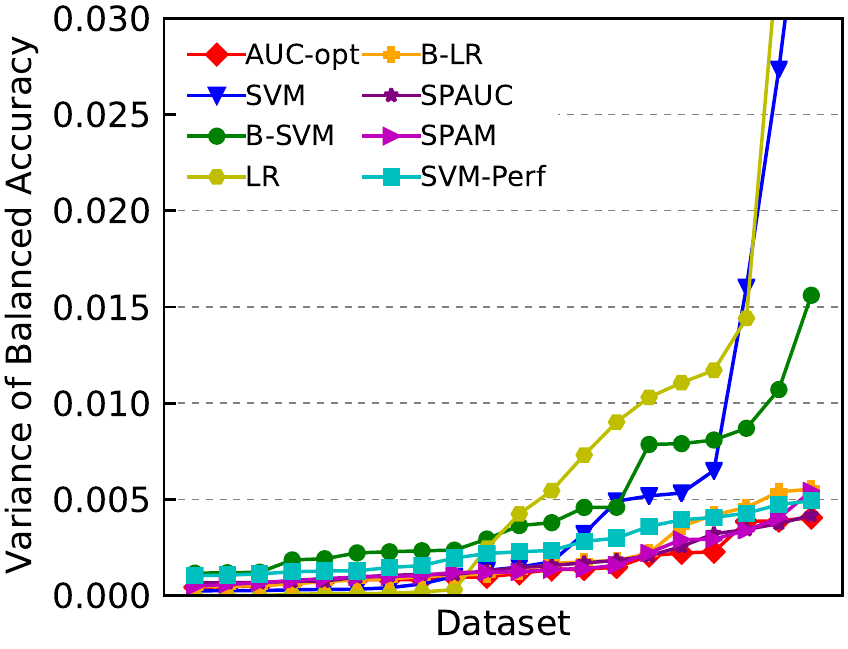}
\end{minipage}\qquad\qquad
\begin{minipage}{0.45\textwidth}
\includegraphics[scale=.8]{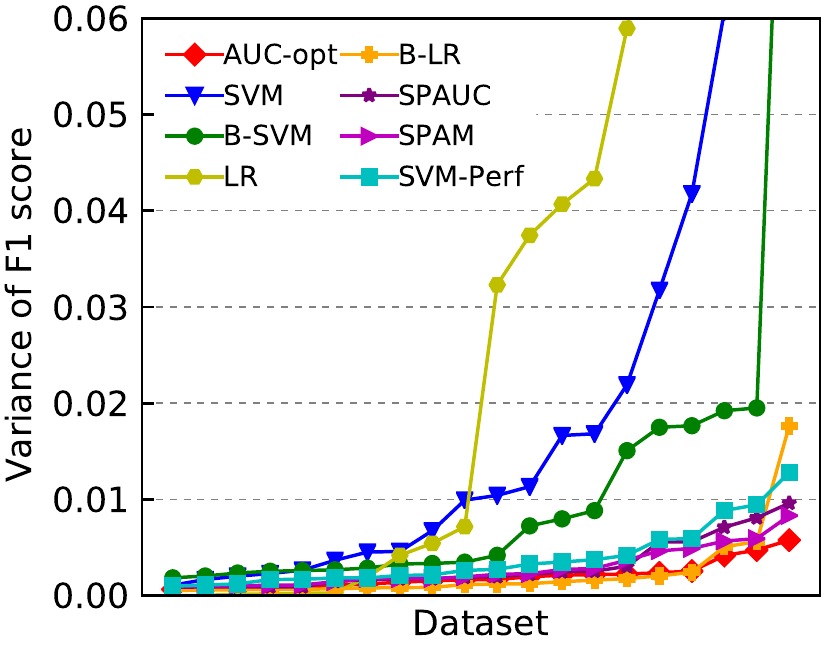}
\end{minipage}
\caption{Variance of balanced accuracies (Left) and F1 scores (Right) as a function of dataset on testing datasets.}
\label{fig:test-b-f1-score}
\end{figure}

\begin{center}
\centering
\def\arraystretch{1.05}
\begin{table}
\centering
\caption{The averaged run time of all classifiers on training stage over all 50 datasets in $\mathbb{R}^2$.}
\begin{tabular}{P{0.15\textwidth}|P{0.25\textwidth}|P{0.15\textwidth}|P{0.25\textwidth}}\toprule
Method & Run Time (mean$\pm$std) & Method & Run Time (mean$\pm$std)\\\hline 
SVM & 0.2186$\pm$0.2825 & RF & 2.9737$\pm$2.5581\\
B-SVM & 0.1830$\pm$0.1811 & B-RF & 2.8366$\pm$2.3977\\ 
LR & 0.1337$\pm$0.1318 & GB & 3.4728$\pm$6.4346\\
B-LR & 0.1345$\pm$0.1357 & SVM-RBF & 3.3854$\pm$17.070\\
SPAM & 3.1661$\pm$7.1141 & B-SVM-RBF & 2.1731$\pm$9.1415\\ 
SPAUC & 3.1996$\pm$7.1521 & RankBoost & 3.9690$\pm$10.169\\ 
SVM-Perf & 7.0858$\pm$9.3885 & AdaBoost & 2.1076$\pm$2.2173\\\bottomrule
\end{tabular}
\label{tab:run-time-real}
\end{table}
\end{center}

\section{More experimental results}
\label{section:experiments-more}
\paragraph{Variances of balanced accuracy and F1 score.} Fig. \ref{fig:test-b-acc-score} presents the balanced accuracy and F1-score on testing datasets. Fig. \ref{fig:test-b-f1-score} illustrates the variance of balanced accuracies and F1 scores of 8 linear methods over 50 t-SNE datasets. The overall performance of these two metrics are consistent with the variance of AUC scores. However, the LR and B-LR have higher variances on balanced accuracies and F1 scores than on AUC scores. This may suggest decision functions of LR and B-LR needs to be calibrated.

\paragraph{Comparisons of training AUC scores.} Table \ref{tab:auc-scores-real-tr} presents the comparison of training AUC scores over 200 trials of 50 real-world datasets. The performance is roughly consistent with results in the testing stage where the nonlinear methods are significantly better than linear methods. The balance weight strategy is clearly improving the linear methods such as SVM and LR but has no improvement for nonlinear methods.

\definecolor{Gray}{gray}{0.9}
\makeatletter

\newcolumntype{"}{@{\hskip\tabcolsep\vrule width .5pt\hskip\tabcolsep}}
\makeatother
\begin{table}[H]
\caption{The comparison of \textit{training} AUC scores over 200 trials of 50 real-world datasets. The value of each cell (I,J) means the number of datasets where method I is significantly better than method J by using a t-test with a significance level of 5\%. Numbers in red region are the results of AUC optimizers better than standard classifiers, while numbers in blue region are the reverse.}
\centering
\begin{tabular}{p{0.12\textwidth}P{0.05\textwidth}P{0.05\textwidth}P{0.05\textwidth}P{0.05\textwidth}P{0.05\textwidth}P{0.05\textwidth}P{0.05\textwidth}P{0.05\textwidth}P{0.05\textwidth}P{0.05\textwidth}P{0.05\textwidth}P{0.05\textwidth}P{0.05\textwidth}P{0.05\textwidth}P{0.08\textwidth}}\toprule 
& \rot{SVM} & \rot{B-SVM} & \rot{LR} & \rot{B-LR} & \rot{RF} & \rot{B-RF} & \rot{GB} & \rot{SVM-RBF} & \rot{\scriptsize{B-SVM-RBF}} & \rot{RankBoost} & \rot{AdaBoost} & \rot{SPAM} & \rot{SPAUC} & \rot{SVM-Perf} & \rot{\textbf{Average}} \\\hline 
SVM   & -  & 2  & 0  & 0  & 0  & 0  & 0  & 4  & 3  & \cellcolor{blue!30!white}0  & \cellcolor{blue!30!white}2  & \cellcolor{blue!30!white}26  & \cellcolor{blue!30!white}37  & \cellcolor{blue!30!white}4 & 6.0 \\
B-SVM   & 29  & -  & 6  & 0  & 0  & 0  & 0  & 4  & 3  & \cellcolor{blue!30!white}1  & \cellcolor{blue!30!white}2  & \cellcolor{blue!30!white}37  & \cellcolor{blue!30!white}42  & \cellcolor{blue!30!white}30 & 11.85 \\
 LR   & 41  & 23  & -  & 4  & 0  & 0  & 0  & 4  & 3  & \cellcolor{blue!30!white}1  & \cellcolor{blue!30!white}3  & \cellcolor{blue!30!white}43  & \cellcolor{blue!30!white}43  & \cellcolor{blue!30!white}43 & 16.0 \\
B-LR   & 42  & 31  & 13  & -  & 0  & 0  & 0  & 4  & 3  & \cellcolor{blue!30!white}1  & \cellcolor{blue!30!white}3  & \cellcolor{blue!30!white}43  & \cellcolor{blue!30!white}44  & \cellcolor{blue!30!white}43 & 17.46 \\
 RF   & 50  & 50  & 50  & 50  & -  & 5  & 32  & 47  & 42  & \cellcolor{blue!30!white}45  & \cellcolor{blue!30!white}48  & \cellcolor{blue!30!white}50  & \cellcolor{blue!30!white}49  & \cellcolor{blue!30!white}50 & 43.69 \\
B-RF   & 50  & 50  & 50  & 50  & 2  & -  & 32  & 46  & 41  & \cellcolor{blue!30!white}45  & \cellcolor{blue!30!white}48  & \cellcolor{blue!30!white}50  & \cellcolor{blue!30!white}49  & \cellcolor{blue!30!white}50 & 43.31 \\
 GB   & 50  & 50  & 50  & 50  & 1  & 1  & -  & 37  & 41  & \cellcolor{blue!30!white}46  & \cellcolor{blue!30!white}47  & \cellcolor{blue!30!white}50  & \cellcolor{blue!30!white}49  & \cellcolor{blue!30!white}50 & 40.15 \\
SVM-RBF   & 45  & 45  & 44  & 44  & 0  & 0  & 6  & -  & 15  & \cellcolor{blue!30!white}21  & \cellcolor{blue!30!white}30  & \cellcolor{blue!30!white}45  & \cellcolor{blue!30!white}44  & \cellcolor{blue!30!white}45 & 29.54 \\
B-SVM-RBF   & 47  & 47  & 46  & 46  & 0  & 0  & 3  & 13  & -  & \cellcolor{blue!30!white}17  & \cellcolor{blue!30!white}27  & \cellcolor{blue!30!white}47  & \cellcolor{blue!30!white}46  & \cellcolor{blue!30!white}47 & 29.69 \\
RankBoost   & \cellcolor{red!30!white}49  & \cellcolor{red!30!white}49  & \cellcolor{red!30!white}46  & \cellcolor{red!30!white}49  & \cellcolor{red!30!white}1  & \cellcolor{red!30!white}1  & \cellcolor{red!30!white}0  & \cellcolor{red!30!white}22  & \cellcolor{red!30!white}26  & -  & 41  & 50  & 49  & 50 & 33.31 \\
AdaBoost   & \cellcolor{red!30!white}48  & \cellcolor{red!30!white}46  & \cellcolor{red!30!white}45  & \cellcolor{red!30!white}43  & \cellcolor{red!30!white}0  & \cellcolor{red!30!white}0  & \cellcolor{red!30!white}0  & \cellcolor{red!30!white}11  & \cellcolor{red!30!white}17  & 1  & -  & 50  & 48  & 49 & 27.54 \\
SPAM   & \cellcolor{red!30!white}12  & \cellcolor{red!30!white}1  & \cellcolor{red!30!white}2  & \cellcolor{red!30!white}0  & \cellcolor{red!30!white}0  & \cellcolor{red!30!white}0  & \cellcolor{red!30!white}0  & \cellcolor{red!30!white}4  & \cellcolor{red!30!white}3  & 0  & 0  & -  & 35  & 15 & 5.54 \\
SPAUC   & \cellcolor{red!30!white}4  & \cellcolor{red!30!white}0  & \cellcolor{red!30!white}1  & \cellcolor{red!30!white}0  & \cellcolor{red!30!white}0  & \cellcolor{red!30!white}0  & \cellcolor{red!30!white}0  & \cellcolor{red!30!white}3  & \cellcolor{red!30!white}3  & 0  & 1  & 6  & -  & 6 & 1.85 \\
SVM-PL   & \cellcolor{red!30!white}6  & \cellcolor{red!30!white}8  & \cellcolor{red!30!white}0  & \cellcolor{red!30!white}0  & \cellcolor{red!30!white}0  & \cellcolor{red!30!white}0  & \cellcolor{red!30!white}0  & \cellcolor{red!30!white}2  & \cellcolor{red!30!white}1  & 0  & 1  & 25  & 37  & - & 6.15 \\\bottomrule
\end{tabular}
\label{tab:auc-scores-real-tr}
\end{table}

\paragraph{Averaged run time over all datasets.}

Table \ref{tab:run-time-real} shows the averaged run time over all 50 real-world datasets. Although different methods may be implemented by using different programming languages, it is clear that these nonlinear methods are much slower than linear methods. Furthermore, the AUC approximate optimizers are generally slower than standard classifiers. This is simply due to the loss of AUC is quadratic with respect to the number of training samples.

\end{document}